\documentclass[conference]{IEEEtran}

\usepackage{epsfig}
\usepackage{amsfonts}
\usepackage{graphicx}
\usepackage{indentfirst}

\usepackage{subfigure}
\usepackage{url}
\usepackage{algorithmic}
\usepackage{algorithm}
\usepackage[cmex10]{amsmath}
\usepackage{amsthm}
\usepackage{array}

\usepackage{fixltx2e}

\usepackage{stfloats}

\usepackage{url}

\newtheorem{thm}{Theorem}[section]

\newtheorem{lem}[thm]{Lemma}

\newtheorem{defn}[thm]{Definition}
\theoremstyle{remark}

\makeatother 
\begin{document}
\title{Competition and Regulation in a Wireless Operators Market: An Evolutionary Game Perspective}

\author{\IEEEauthorblockN{Omer Korcak\IEEEauthorrefmark{1},
George Iosifidis\IEEEauthorrefmark{2},
Tansu Alpcan\IEEEauthorrefmark{3},
and Iordanis Koutsopoulos\IEEEauthorrefmark{2}}
\IEEEauthorblockA{\IEEEauthorrefmark{1}Dept. of Computer Engineering, Marmara University, Turkey}
\IEEEauthorblockA{\IEEEauthorrefmark{2} Dept. of Computer and Comm. Engineering, University of Thessaly and CERTH, Greece}
\IEEEauthorblockA{\IEEEauthorrefmark{3}Dept. of Electrical and Electronic Engineering, University of Melbourne, Australia}}

%\author{Omer Korcak, George Iosifidis, Tansu Alpcan, and Iordanis Koutsopoulos
%\IEEEcompsocitemizethanks{
%\IEEEcompsocthanksitem
%O. Korcak is with the Dep. of Comp. Engineering, Marmara University, Turkey,
%G. Iosifidis and I. Koutsopoulos are with Dep. of Comp. and Comm. Eng., University of Thessaly and CERTH,
%T. Alpcan is with Dep. of Electrical and Elect. Eng., University of Melbourne, Australia} }

\maketitle

\begin{abstract}
We consider a market where a set of wireless operators compete for a large common pool of users. The latter have a reservation utility of $U_0$ units or, equivalently, an alternative option to satisfy their communication needs. The operators must satisfy these minimum requirements in order to attract the users. We model the users decisions and interaction as an evolutionary game and the competition among the operators as a non cooperative price game which is proved to be a potential game. For each set of prices selected by the operators, the evolutionary game attains a different stationary point. We show that the outcome of both games depends on the reservation utility of the users and the amount of spectrum $W$ the operators have at their disposal. We express the market welfare and the revenue of the operators as functions of these two parameters. Accordingly, we consider the scenario where a regulating agency is able to intervene and change the outcome of the market by tuning $W$ and/or $U_0$. Different regulators may have different objectives and criteria according to which they intervene. We analyze the various possible regulation methods and discuss their requirements, implications and impact on the market.
\end{abstract}

\section{Introduction}

Consider a city where $3$ commercial operators (companies) and one
municipal operator offer WiFi Internet access to the citizens
(users). The companies charge for their services and offer better
rates than the municipal WiFi service which however is given gratis.
Users with high needs will select one of the companies. However, if
they are charged with high prices, or served with low rates, a
portion of them will eventually migrate to the municipal network. In
other words, the municipal service constitutes an alternative choice
for the users and therefore sets the minimum requirements which the
commercial providers should satisfy. Apparently, the existence of
the municipal network affects both the user decisions and the
operators pricing policy. In different settings, the minimum
requirement can be an inherent characteristic of the users as for
example a lower bound on transmission rate for a particular
application, an upper bound on the price they are willing to pay or
certain combinations of both of these parameters. Again, the
operators can attract the users only if they offer more appealing
services and prices.

In this paper, we consider a general wireless communication services
market where a set of \emph{operators}, compete to sell their
services to a common large pool of \emph{users}. We assume that
users have minimum requirements or alternative options to satisfy
their needs which we model by introducing the reservation utility
$U_0$, \cite{acemoglou}. Users select an operator only if the
offered service and the charged price ensure utility higher than
$U_0$. We analyze the users strategy for selecting operator and the
price competition among the operators under this constraint. We find
that the market outcome depends on $U_0$ and on the amount of
spectrum each operator has at his disposal $W$. Accordingly, we
consider the existence of a regulating agency who is interested in
affecting the market and enforcing a more desirable outcome, by
tuning either $W$ or $U_0$. For example, consider the municipal WiFi
provider who is actually able to set $U_0$ and bias the competition
among the commercial providers. This is of crucial importance since
in many cases the competition of operators may yield inefficient
allocation of the network resources, \cite{acemoglou} or even
reduced revenue for them, \cite{walrand}. We introduce a rigorous
framework that allows us to analyze the various methods through
which the regulator can intervene and affect the market outcome
according to his objective.

\begin{figure}[t]
\begin{center}
\epsfig{figure=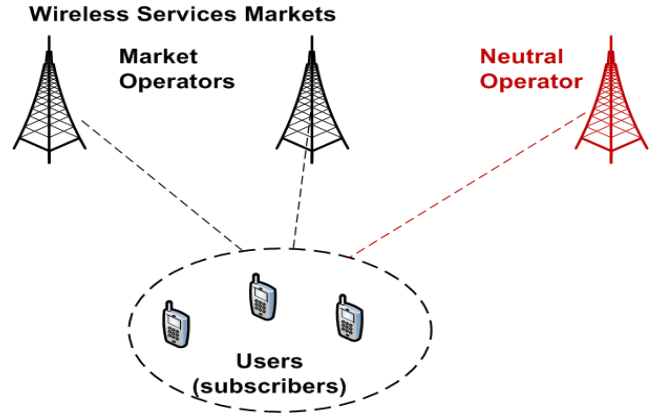, width=8.5cm,height=5.5cm}
\end{center}
\caption[An Example of Wireless Services Market]{The market consists
of a set of operators competing over a common pool of users. Each user selects one of the market operators or opts to abstain from the market and be associated with the neutral operator. The latter models the alternative out-of-the-market users option or their minimum service-price requirements.} \label{omerfig:exampleMarket}
\end{figure}

Our model captures many different settings such as a WiFi market in
a city, a mobile/cell-phone market in a country or even a secondary
spectrum market where primary users lease their spectrum to
secondary users. In order to make our study more realistic, we adopt
a macroscopic perspective and analyze the interaction of the
operators and users in a large time scale, for large population of
users, and under limited information. The operators are not aware of
the users specific needs and the latter cannot predict in advance
the exact level of service they will receive. Each operator has a
total resource at his disposal (e.g. the aggregate service rate)
which is on average equally allocated to his subscribers,
\cite{acemoglou}, \cite{berry}. This is due to the various network
management and load balancing techniques that the operators employ,
or because of the specific protocol that is used, \cite{lin}. Each
user selects the operator that will provide the optimal combination
of service quality and price. Apparently, the decision of each user
affects the utility of the other users. We model this
interdependency as an evolutionary game, \cite{sandholm1} the
stationary point of which represents the users distribution among
the operators and depends on the charged prices. This gives rise to
a non cooperative price competition game among the operators who
strive to maximize their profits.

Central to our analysis is the concept or the \emph{neutral
operator} $P_0$ which provides to the users a constant and given
utility of $U_0$ units. The $P_0$ can be a special kind of operator,
like the municipal WiFi provider in the example above, or it can
simply model the user choice to abstain from the market. This way,
we can directly calculate how many users are served by the market
and how many abstain from it and select $P_0$. Moreover, $P_0$
allows us to introduce the role of a regulating agency who can
intervene and bias the market outcome through the service $U_0$. We
show that $P_0$ can be used to increase the revenue of the operators
or the efficiency of the market. In some cases, both of these
metrics can be simultaneously improved at a cost which is incurred
by the regulator. Alternatively, the outcome of the market can be
regulated by changing the amount of spectrum each operator has at
his disposal. Different regulating methods give different results
and entail different cost for the regulator.

\begin{figure}[t]
\begin{center}
\epsfig{figure=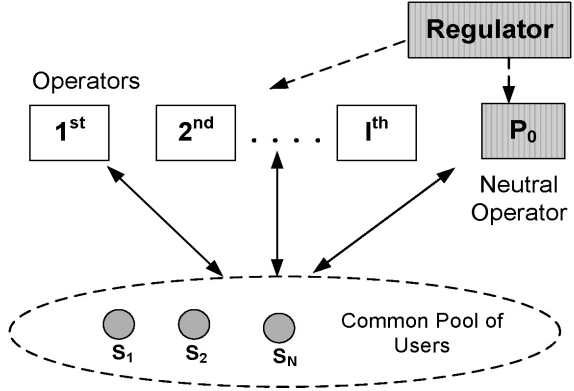, width=8cm,height=5.4cm}
\end{center}
\caption[Oligopolistic Wireless Services Market - System Model]{The
oligopoly market consists of $I$ operators and $N$ users (S). Each user is
associated with one operator at each specific time slot. Every
operator $i=1,2,\ldots,I$ can serve more than one users at a certain
time slot. The users that fail to satisfy their minimum
requirements, $U_i\leq U_0$, $\forall i\in\mathcal{I}$, abstain from
the market and select the neutral operator $P_0$.}
\label{omerfig:MarketModel}
\end{figure}

\subsection{Related Work and Contribution}

The competition of sellers for attracting buyers has been studied
extensively in the context of network economics,
\cite{courcoubetis}, \cite{srikant_economics}, both for the Internet and more
recently for wireless systems. In many cases, the competition
results in undesirable outcome. For example, in \cite{acemoglou} the
authors consider an oligopoly communication market and show that it
yields inefficient resource allocation for the users. From a
different perspective it is explained in \cite{walrand}, that
selfish pricing strategies may also decrease the revenue of the
sellers-providers. In these cases, the strategy of each node (buyer)
affects the performance of the other nodes by increasing the delay
of the services they receive, \cite{acemoglou} (\emph{effective
cost}) or, equivalently, decreasing the resource the provider
allocates to them, \cite{berry} (\emph{delivered price}). This
equal-resource sharing assumption represents many different access
schemes and protocols (TDMA, CSMA/CA, etc), \cite{lin}.

More recently, the competition of operators in wireless services
markets has been studied in \cite{berry}, \cite{huang1},
\cite{niyato1}, \cite{niyato2}, \cite{tuffin}. The users can be
charged either with a usage-based pricing scheme, \cite{huang1}, or
on a per-subscription basis, \cite{niyato1}, \cite{tuffin}. We adopt
the latter approach since it is more representative of the current
wireless communication systems. We assume that users may migrate
(churn) from one operator to the other, \cite{tuffin}, and we use
evolutionary game theory (EVGT) to model this process,
\cite{niyato-churn}. This allows us to capture many realistic
aspects and to analyze the interaction of very large population of
users under limited information. The motivation for using EVGT in
such systems is very nicely discussed in \cite{niyato2}. Due to the
existence of the \emph{neutral operator}, the user strategy is
updated through a hybrid scheme based on imitation and direct
selection of $P_0$. We define a new revision protocol to capture
this aspect and we derive the respective system dynamic equations.

Although the regulation has been discussed in context of networks,
\cite{courcoubetis}, it remains largely unexplored. Some recent work
\cite{huang-icc}, \cite{park} study how a regulator or an
\emph{intervention device} may affect a non-cooperative game among a
set of players (e.g. operators). However, these works do not
consider hierarchical systems, with large populations and limited
information. Our contribution can be summarized as follows:
\textbf{(i)} we model the wireless service market using an
evolutionary game where the users employ a new hybrid revision protocol,
based both on imitation and direct selection of a specific choice, namely the $P_0$. We derive the
differential equations that describe the evolution of this system
and find the stationary points, \textbf{(ii)} we define the price
competition game for $I$ operators and the particular case that
users have minimum requirements, or equivalently, alternative
choices/offers, \textbf{(iii)} we prove that this is a Potential
game and we analytically find the Nash equilibria, \textbf{(iv)} we
introduce the concept of the neutral operator who represents the
system/state regulator or the minimum users requirement, and
\textbf{(v)} we discuss different regulation methods and analyze
their efficacy, implications and the resources that are required for
their implementation.

The rest of this paper is organized as follows. In Section
\ref{Omersec:section2 - system model} we introduce the system model
and in Section \ref{Omersec:section3 - service market dyn} we
analyze the dynamics of the users interaction and find the
stationary point of the market. In Section \ref{Omersec:section4 -
operators_competition} we define and solve the price competition
game among the operators and in Section \ref{Omersec:section5 -
Regulation} we discuss the relation between the revenue of the
operators and the efficiency of the market and their dependency on
the system parameters. Accordingly, we analyze various regulation
methods for different regulation objectives and give related
numerical examples. We conclude in Section \ref{Omersec:section6 -
conclusions}.

\section{System Model} \label{Omersec:section2 - system model}
We consider a wireless service market (hereafter referred to as a
\emph{market}) with a very large set of users
$\mathcal{N}=(1,2,\ldots,N)$ and a set of operators
$\mathcal{I}=(1,2,\ldots,I)$, which is depicted in Figure
\ref{omerfig:MarketModel}. We assume a time slotted operation. Each user cannot be served by more than one operator simultaneously. However, users can switch in each slot $t$ between operators or even they can opt to refrain and not purchase services
from anyone of the $I$ operators. The net utility perceived by each
user who is served by operator $i$ in time slot $t$ is:
\begin{equation}
U_i(W_i, n_i(t),\lambda_i)=V_i(W_i,n_i(t))-\lambda_i
\label{omereq:utility_Ui}
\end{equation}
where $n_i(t)$ are the users served by this specific
operator in slot $t$, $W_i$ the total spectrum at his disposal, and $\lambda_i$
the charged price. In order to describe the market operation we
introduce the users vector
$\mathbf{x}(t)=(x_1(t),x_2(t),\ldots,x_I(t),x_0(t))$, where the
$i^{th}$ component $x_i(t)=n_i(t)/N$ represents the portion of
users that have selected operator $i\in\mathcal{I}$. Additionally,
with $x_0(t)=n_{0}(t)/N$ we denote the portion of users that have
selected neither of the $I$ operators. We assume that the number of users $N$ is very large, $N>>1$ and therefore the variable $x_i(t)=n_i(t)/N$ is considered continuous. In other words, we assume that there exist a continuum of users partitioned among the different operators.

\subsubsection{Valuation function}

The function $V_i(\cdot)$ represents the value of the offered service for each user associated with
operator $i\in\mathcal{I}$. Users are considered homogeneous: all the users served by a certain operator are
charged the same price and perceive the same utility. We consider the following
particular valuation function:
\begin{equation}
V_i(W_i,x_i(t))=\log{\frac{W_i}{Nx_i(t)}},\,x_{i}(t)>0 \label{omereq:log_valuation}
\end{equation}
Since $N$ is given, we use $x_i(t)$ instead of $n_i(t)$.
This function has the following desirable properties: \textbf{(i)} the valuation
for each user decreases with the total number of served users by the specific operator due to congestion, \textbf{(ii)}
increases with the amount of available spectrum $W_i$, and \textbf{(iii)} it is a concave function and
therefore captures the saturation of the user satisfaction as the allocated resource increases, i.e. it satisfies the
principle of diminishing marginal returns, \cite{courcoubetis}.

A basic assumption in our model is that users served by the same operator are allocated an equal amount of resource. We want to stress that this assumption captures many different settings in wireline, \cite{acemoglou},
or wireless networks, \cite{berry}, \cite{lin}, \cite{niyato1},
\cite{niyato2}, \cite{XuChen}. Some examples where the equal resource sharing assumption holds are the following:
\begin{itemize}

    \item \textbf{FDMA - TDMA:} If the operator uses a multiple access scheme like Frequency Division Multiple Access (FDMA) or Time Division Multiple Access (TDMA), then the equal resource sharing assumption holds by default, \cite{lin}. The users served by a certain operator receive an equal share of his total available spectrum or an equal time share of the operator's channel.

    \item \textbf{CSMA/CA:} A similar result holds for the Carrier Sense Multiple Access scheme with Collision Avoidance, \cite{bianchi}, that is used in IEEE 802.11 protocols. Users trying to access the channel receive an equal share of it and achieve - on average - the same transmission rate. Additionally, as it was shown in \cite{cagalj}, even if the radio transmitters are controlled by selfish users, they can achieve this fair resource sharing.

    \item \textbf{Random access of multiple channels:} Even in more complicated access schemes as in the case, for example, where many different users iteratively select the least congested channel among a set of available channels, it is proved that each user receives asymptotically an equal share of the channel bandwidth, \cite{XuChen}.

\end{itemize}
Additionally, the macroscopic perspective and the large time scale that we consider in this problem , ensure that spatiotemporal variations in the quality of the offered services will be smoothed out due to load balancing and other similar network management techniques that the operators employ. Therefore, users of each operator are treated in equal terms.

\subsubsection{Neutral Operator}

Variable $x_0(t)$ represents the portion of users that do not select anyone of the $I$ operators. Namely, a user
in each time slot $t$ is willing to pay operator $i\in\mathcal{I}$ only if the offered utility $U_i(W_i,x_i(t),\lambda_i)$ is
greater than a threshold $U_0\geq 0$. If all operators fail to satisfy this minimum requirement then the user abstains from the
market and is associated with the \emph{Neutral Operator} $P_0$,
Figure \ref{omerfig:MarketModel}. In other words, $P_0$ represents
the choice of selecting neither of the $I$ operators and receiving
utility of $U_0$ units. Technically, as it will be
shown in the sequel, the inclusion of $P_0$ affects both the user
decision process for selecting operator and the competition among
the operators.

From a modeling perspective, the neutral operator may be used to
represent different realistic aspects of the wireless service
market. First, $P_0$ can be an actual operator owned by the state,
as the public/municipal WiFi provider we considered in the
introductory example. In this case, through the gratis $U_0$
service, the state intervenes and regulates the market as we will
explain in Section \ref{Omersec:section5 - Regulation}.
Additionally, $U_0$ can be indirectly imposed by the state (the
regulator) through certain rules such as the minimum amount of
spectrum/rate per user. Finally, it can represent the users
reluctancy to pay very high prices for poor QoS, similarly to the
individual rationality constraint in mechanism design. We take these
realistic aspects into account and moreover, by using $x_0(t)$, we
find precisely how many users are not satisfied by the market of the
$I$ operators.

Unlike the valuation $V_i(\cdot)$ of the service offered by each operator $i\in\mathcal{I}$, $U_0$ is considered constant. When $U_0$ represents users minimum requirements or respective restrictions imposed by regulatory rules, this assumption
follows directly and actually is imperative. In case $U_0$ models the service offered by the neutral operator (e.g. the municipal WiFi network), the constant value of $U_0$ means that it is independent of the number of users and hence non-congestible. We follow this assumption for the following two reasons:\textbf{(i)} $U_0$ is a free of charge service which in general is low and hence can be ensured for a large number of users. \textbf{(ii)} The state agency (i.e. the regulator) who provides $U_0$, is able to increase his resource in order to ensure a constant value for $U_0$. As we will explain in next sections, this latter aspect captures the cost of regulation, i.e. the cost of serving users through the neutral operator. Finally, notice that our model can be easily extended for the case that $U_0$ is a congestible service.

\subsubsection{Revenue}

Each operator $i\in\mathcal{I}$ determines the price $\lambda_i\in R^{+}$ that he will charge to his clients. The decisions of the operators are realized in a different time scale than the decisions of the users. Namely, each operator $i$ determines his price in the beginning of each time epoch $\mathcal{T}$ which consists of $T$ slots, while users update their operator association decision in each slot. Let us define the price vector $\mathbf{\lambda}=(\lambda_i\,:\,i=1,2,\ldots,I)$ and the
vector of the $I-1$ prices of operators other than $i$ as
$\mathbf{\lambda}_{-i}=(\lambda_j\,:\,j\in\mathcal{I}\setminus{i})$.
We assume that $T$ is large enough so that for each price vector $\mathbf{\lambda}$ set at the beginning of an epoch, the market of the users reaches a stationary point - if such a point is attainable - during this epoch. The objective of each operator $i\in\mathcal{I}$ is to maximize his revenue during each epoch $\mathcal{T}$:
\begin{equation}
R_i(x_i(t),\lambda_i)=\lambda_ix_i(t)N
\label{omereq:Operator_Revenue}
\end{equation}
In these markets there are no service level agreements (SLAs) or any
other type of QoS guarantees and hence the operators are willing to
admit and serve as many users as it is required to achieve their goal.

%%-------------------- possible description ----------------- ::
%%-----------------------------------------------------------
%the competition between wireless access providers, where an outside option (a free %alternative, that may consist in not subscribing to any service) is available. The %authors study that situation as a multi-level noncooperative game, where at a small %time scale users select their preferred provider based on a quality-price tradeoff, %at a larger time scale providers set their price so as to maximize revenue, and an %even larger time scale can be considered, that consist in a regulator (or a public %entity) fixing the amount of spectrum sold to providers and/or the quality of the %free option

\section{User Strategy and Market Dynamics} \label{Omersec:section3 - service market dyn}

\subsection{Evolutionary Game $\mathcal{G_U}$ among Users}
In order to select the optimal operator that maximizes eq.
(\ref{omereq:utility_Ui}), each user must be aware of all system
parameters, i.e. the spectrum $W_i$, the number of served users
$n_i$ and the charged price $\lambda_i$ for each $i\in\mathcal{I}$.
However, in realistic settings this information will not be
available in advance. Given these restrictions and the large number
of users, we model their interaction and the operator selection
process by defining an evolutionary game, $\mathcal{G_U}$, as
follows:
\begin{itemize}
\item Players: the set of the $N$ users, $\mathcal{N}=(1,2,\ldots,N)$.
\item Strategies: each user selects a certain operator $i\in\mathcal{I}$ or the neutral
operator $P_0$.
\item Population State: the users distribution over the $I$
operators and the neutral operator,
$\mathbf{x}(t)=(x_1(t),x_2(t),\ldots,x_{I}(t),x_{0}(t))$.
\item Payoff: the user's net utility $U_i(W_i, x_i(t),\lambda_i)$
when he selects operator $i\in\mathcal{I}$, or $U_0$ when he selects
$P_0$.
\end{itemize}

To facilitate our analysis we make the following assumptions:
\begin{itemize}
\item \textbf{Assumption 1}: The number of users $N$ is very large, $N>>1$ and therefore the variable $x_i(t)=n_i(t)/N$ is considered continuous.
\item \textbf{Assumption 2}: The initial distribution of users over the $I$ operators is non zero: $x_i(0)>0$, $\forall i\in\mathcal{I}$. It directly follows that $x_0(0)<1$.
\end{itemize}
In the sequel we explain how each user selects his strategy under this limited information and what is the outcome of this game.

\subsection{User Strategy Update}
A basic component of every evolutionary game is the \emph{revision protocol}, \cite{sandholm1}. It captures the dynamics of the interaction among the users and describes in detail the process according to which a player iteratively updates his strategy. There exist many different options for the revision protocol, depending on the modeling assumptions of the specific problem. These assumptions are mainly related to how sophisticated, informed and rational are the players. On the one extreme, fully rational and informed players update their choices according to a best response strategy like in the typical (non-evolutionary) strategic games. This means that players make a \emph{direct selection} of the best available strategy. On the other extreme, players follow an imitation strategy. In this case a player $(A)$ selects randomly another player $(B)$ and if the utility of the latter is higher, $(A)$ imitates his strategy with a probability that is proportional to the anticipated utility improvement. This modeling option is suitable for imperfectly informed players, or players with bounded rationality who update their strategy based on a better (instead of best) response strategy. Between these two extremes, there are many different options. For example, a player may update his strategy with a \emph{hybrid} protocol based partially on imitation and on direct selection, \cite{sandholm1}.

In this work, we assume that each user updates his strategy by a special type of hybrid revision protocol which is a combination of imitation of other users associated with operators from the set $\mathcal{I}$ (market operators) and direct selection of the neutral operator $P_0$. The imitation component captures the lack of information users have at their disposal about the market. On the other hand, each user is aware of the exact value of $U_0$ and hence this choice is always available through direct selection. Notice that the considered revision protocol is not a typical hybrid protocol since the direct selection is related only to the selection of $P_0$ and not to the other operators.

In detail, the proposed revision protocol can be described by the following actions that each user may take in each slot $t$:
\begin{enumerate}
\item A user associated with an operator $i\in\mathcal{I}$, selects randomly another user who is associated with an operator $j\in\mathcal{I},\,j\neq i$, and if $U_j>U_i$ imitates his strategy with a probability that is proportional to the difference $(U_j-U_i)$.
\item A user associated with the neutral operator $P_0$, selects randomly another user associated with operator $j\in\mathcal{I}$ and if $U_j>U_0$, imitates his strategy with a probability that is proportional to the difference $(U_j-U_0)$.
\item A user associated with operator $i\in\mathcal{I}$ selects the neutral operator $P_0$ with probability that is proportional to the difference $(U_0-U_i)$.
\end{enumerate}
Options $1$ and $2$, stem from the replicator dynamics introduced by Taylor and Jonker in \cite{taylor} and are based on imitation of users with better strategies. On the other hand, option $3$ is based on direct selection of better strategies, known also as pairwise dynamics, introduced by Smith in \cite{smith}.

After defining the revision protocol, we can calculate the rate at which users switch from one strategy (operator) to another strategy (operator). In particular, the switch rate of users migrating from operator $i$ to operator $j\in\mathcal{I}\setminus{i}$ in time slot $t$, is:
\begin{equation}
\rho_{ij}(t)=x_j(t)[U_j(t)-U_i(t)]_{+} \label{omereq:RevisionforOper}
\end{equation}
where $x_j(t)$ is the portion of users already associated with operator $j$. For simplicity, we express the user utilities as a function with a single argument, the time $t$.
Additionally, the users switch rate from operator $i$ to neutral operator $P_0$, is:
\begin{equation}
\rho_{i0}(t)=\gamma[U_0-U_i(t)]_{+}
\label{omereq:RevisionwithArtifOper}
\end{equation}
Notice the difference between imitation and direct selection
\cite{sandholm1}. Instead of multiplying the utilities difference with the population $x_0(t)$, we use a constant multiplier $\gamma\in R$. This is due to the model assumption that switching to the neutral operator is not accomplished through imitation and hence does not depend on the portion of users already been associated with $P_0$.  The probabilistic aspect captures the bounded rationality, the inertia of the users and other similar realistic
aspects of these markets. Finally, the switch rate of users leaving $P_0$ and returning to the market (option $2$) is:
\begin{equation}
\rho_{0i}(t)=x_i(t)[U_i(t)-U_0]_{+}
\end{equation}

Variables $\rho_{ij}$, $\rho_{i0}$ and $\rho_{0i}$ represent the rates at which users migrate from one operator to another, including the neutral operator $P_0$. It is interesting to notice that if these rates are normalized properly, they can be interpreted as the probabilities with which users update their operator selection strategy. This approach is discussed in \cite{sandholm1}. In the sequel we use these rates to derive the ordinary differential equations (ODE) that describe the evolution of the population of users.

\subsection{Market Stationary Points}
The new type of hybrid revision protocol introduced above, results in user market dynamics that cannot be expressed with the known differential equations of replicator dynamics or other similar scheme, \cite{sandholm1}. In Section \ref{appendixproof:derivofevolutionarydynamics} of the Appendix we prove that the mean dynamics of the system are:
\begin{eqnarray}
\frac{d x_i(t)}{dt}&=&x_i(t)[U_i(t)-U_{avg}(t)-x_0(t)(U_i(t)-U_0)\\ \nonumber &-&\gamma
(U_0-U_i(t))_{+}+x_0(t)(U_i(t)-U_0)_{+}],\,\forall i\in\mathcal{I}
\end{eqnarray}
where $U_{avg}(t)=\sum_{i\in\mathcal{I}}x_i(t)U_i(t)$ is the average
utility of the market in each slot $t$. The user population
associated with $P_0$ is:
\begin{equation}
\frac{d x_0(t)}{dt}=x_0 \sum_{i \in
\mathcal{I}^{+}}{x_i(U_0-U_i)} + \gamma\sum_{j \in
\mathcal{I}^{-}}{x_j(U_0-U_j)} \label{omereq:x-0-dot}
\end{equation}
where $\mathcal{I}^{+}$ is the subset of operators offering utility
$U_i(t)>U_0$, and $\mathcal{I}^{-}$ is the subset of operators
offering utility $U_i(t)<U_0$, at slot $t$.

The important thing is that despite its different
evolution, as we prove in Section \ref{appendixproof:AnalysisOfStationaryPoints},
this system has the same stationary points as the systems that are described by the classical replicator dynamic equations:
\begin{equation}
\Dot{x}_i(t)=0 \Rightarrow
x_i(t)[U_i(t)-U_{avg}(t)]=0,\,\forall i\in\mathcal{I}
\label{omereq:Pi_stationarypoints}
\end{equation}
and
\begin{equation}
\Dot{x}_0(t)=0\Rightarrow x_0(t)[U_0-U_{avg}(t)]=0
\label{omereq:P0_stationarypoints}
\end{equation}
The user state vector $\mathbf{x}^{*}$ and the respective user
utility $U_{i}^{*}$, $i\in\mathcal{I}$, that satisfy these
stationary conditions can be summarized in the following $3$ cases:
\begin{itemize}
\item \textbf{Case A}: $x_{i}^{*}$, $x_{0}^{*}>0$ and
$U_{i}^{*}=U_0$, $i\in\mathcal{I}$.
\item \textbf{Case B}: $x_{i}^{*}$, $x_{j}^{*}>0$, $x_{0}^{*}=0$ and
$U_{i}^{*}=U_{j}^{*}$, with $\,$
$U_{i}^{*},\,U_{j}^{*}>U_0$,$\,\forall\,i,\,j\in\mathcal{I}$.
\item \textbf{Case C}: $x_{i}^{*}$, $x_{j}^{*}>0$, $x_{0}^{*}=0$ and
$U_{i}^{*}=U_{j}^{*}=U_0$, $\,\forall\,i,\,j\in\mathcal{I}$.
\end{itemize}
Case $A$ corresponds to the scenario where all operators offer to their clients net utility which is equal to the value of the service offered by the neutral operator. On the other hand, in case $B$ the market operators offer higher utility than the neutral operator and hence all users are served by the market. Finally, in case $C$, the $I$ operators offer marginal services, i.e. equal to $U_0$, but they have attracted all the users.

It is interesting to compare the above results with the Wardrop model and the Wardrop equilibrium, \cite{wardrop}. The market stationary points for \textbf{Case A} and \textbf{Case C} satisfy the \emph{Wardrop first principle} and yield an equilibrium where the available strategy options ("operators" in our problem) result in equal utility for the players ("users"). However, this does not hold for \textbf{Case B} where operators other than $P_0$ offer higher utility. This emerges due to the fact that the alternative option (or reservation utility) is non-congestible, i.e. independent of $x_0$. The evolutionary game allows us to provide a richer model than the typical Wardrop model and more importantly to capture the users interaction and dynamics.

Before calculating the stationary point $\mathbf{x}^{*}$ for each
case, and in order to facilitate our analysis, we define the
scalar parameter $\alpha_i=W_i/(Ne^{U_0})$ for each operator $i\in\mathcal{I}$ and the respective vector
$\mathbf{\alpha}=(\alpha_i\,:\,i=1,2,\ldots,I)$. As it will be
explained in the sequel, these parameters determine the operators
and users interaction and will help us to explain the role of the regulator. We can find the stationary points for \textbf{Case A} by
using equation $U_{i}(W_i,x_{i}^{*},\lambda_i)=U_0$ and imposing
the constraint $x_{0}^{*}>0$. Apparently, the state vector
$\mathbf{x}^{*}$ depends on the price vector $\mathbf{\lambda}$.
Therefore, we define the set of all possible \textbf{Case A}
stationary points, $X_A$, as follows (see Section \ref{appendixproof:AnalysisOfStationaryPoints} for
details):
\begin{equation}
X_A=\left\{x_{i}^{*}=\alpha_ie^{-\lambda_i},\forall
i\in\mathcal{I},
x_{0}^{*}=1-\sum_{i=1}^{I}\alpha_ie^{-\lambda_i}:\,\mathbf{\lambda}\in
\Lambda_A\right\}\nonumber
\end{equation}
where $\Lambda_A$ is the set of prices for which a stationary point
in $X_A$ is attainable, i.e. for which it holds $x_{0}^{*}>0$:
\begin{equation}
\Lambda_A=\left\{
(\lambda_1,\lambda_2,\ldots,\lambda_I)\,:\,\sum_{i=1}^{I}\alpha_ie^{-\lambda_i}<1
\right\} \nonumber
\end{equation}
Recall that due to the very large number of users, we consider $x_i$ a continuous variable.

Similarly, for \textbf{Case B}, we calculate the stationary points
by using the set of equations
$U_{i}(W_i,x_{i}^{*},\lambda_i)=U_{j}(W_j,x_{j}^{*},\lambda_j)$,
$\forall\,i,j\in\mathcal{I}$:
\begin{equation}
X_B=\left\{x_{i}^{*}=\frac{\alpha_i}{e^{\lambda_i}\sum_{j=1}^{I}\alpha_je^{-\lambda_j}},\forall
i\in\mathcal{I},x_{0}^{*}=0:\,\mathbf{\lambda}\in \Lambda_B\right\}\nonumber
\end{equation}
where $\Lambda_B$ is the set of prices for which a stationary point
in $X_B$ is feasible, i.e. $U_{i}^{*}>U_0$:
\begin{equation}
\Lambda_B=\left\{
(\lambda_1,\lambda_2,\ldots,\lambda_I)\,:\,\sum_{i=1}^{I}\alpha_ie^{-\lambda_i}>1
\right\}\nonumber
\end{equation}
Finally, the stationary points for the \textbf{Case C} solution must
satisfy the constraint $\sum_{i=1}^{I}\alpha_ie^{-\lambda_i}=1$
which yields:
\begin{equation}
X_C=\left\{x_{i}^{*}=\alpha_ie^{-\lambda_i},\forall\,i\in\mathcal{I},x_{0}^{*}=0:\,\mathbf{\lambda}\in
\Lambda_C\right\}\nonumber
\end{equation}
with
\begin{equation}
\Lambda_C=\left\{
(\lambda_1,\lambda_2,\ldots,\lambda_I)\,:\,\sum_{i=1}^{I}\alpha_ie^{-\lambda_i}=1
\right\}\nonumber
\end{equation}
Notice that the stationary point sets $X_A$, $X_B$ and $X_C$ and the
respective price sets, $\Lambda_A$, $\Lambda_B$, and $\Lambda_C$
depend on the vector $\mathbf{\alpha}$. These results are summarized
in Table \ref{omertable:statpoints}. For each operators price profile $\mathbf{\lambda}$, the evolutionary game admits a unique stationary point $\mathbf{x}^{*}=(x_{1}^{*},x_{2}^{*},\ldots,x_{I}^{*},x_{0}^{*})$ which belongs in the respective set $X_A$, $X_B$, or $X_C$. The utility of the users is
equal to $U_0$ for the \textbf{Case A} and \textbf{Case C}, while
for \textbf{Case B} it depends on $\mathbf{\lambda}$.

\begin{table}[t]
\caption[Wireless Service Market Stationary Points]{Wireless Service Market Stationary Points.} % title of Table
\centering % used for centering table
\begin{tabular}{c c c c} % centered columns (4 columns)
\hline\hline %inserts double horizontal lines
  & $X_A$ & $X_B$ & $X_C$  \\ [0.9ex] % inserts table
%heading
\hline % inserts single horizontal line
$x_{i}^{*}$ & $\alpha_ie^{-\lambda_i}$ &$\frac{\alpha_i}{e^{\lambda_i}\sum_{j=1}^{I}\alpha_je^{-\lambda_j}}$ & $\alpha_ie^{-\lambda_i}$  \\ %
$x_{0}^{*}$ & $1-\sum_{i=1}^{I}\alpha_ie^{-\lambda_i}$ & $0$ & $0$  \\%
Cond. & $\mathbf{\lambda}\in\Lambda_A$ & $\mathbf{\lambda}\in\Lambda_B$ & $\mathbf{\lambda}\in\Lambda_C$ \\%
\hline%inserts single line
\end{tabular}
\label{omertable:statpoints} % is used to refer this table in the text
\end{table}

\subsubsection{Stability of Stationary Points}

Now that we found the stationary points of the hybrid revision protocol, it is important to characterize their stability. We prove in the sequel that these points are Evolutionary Stable Strategies (ESS) and hence they are locally asymptotically stable, i.e. stable within a limited region. ESS and replicator dynamics are the two concepts used for studying evolutionary games. Unlike the replicator dynamics, ESS is a static concept which requires that the strategy of players in the equilibrium is stable when it is invaded by a small population of players playing a different strategy, \cite{levin-learning}. When the players population is homogeneous, as we assumed in our model, an ESS is stable in the replicator dynamic, but not every stable steady state is an ESS. Additionally, every ESS is Nash, and hence ESS is a refinement of the Nash equilibrium.

Let us first give a simple definition of the ESS, tailored to our system model. Assume that the users market has reached the stationary state described by vector $\mathbf{x}^{*}$. Suppose now that a small portion $\epsilon>0$ of the users population deviates from their decision in the stationary state (i.e. selects another operator) and selects another operator $j\in\mathcal{I}$ or the neutral operator. This yields a new distribution of users which we denote by $\mathbf{x}_{\epsilon}=(x_{1}^{\epsilon},x_{2}^{\epsilon},\ldots,x_{I}^{\epsilon},x_{0}^{\epsilon})$. We say that $\mathbf{x}^{*}$ is an ESS if (i) users that deviate from $\mathbf{x}^{*}$ receive lower utility in the new system state $\mathbf{x}_{\epsilon}$ or, (ii) the utility of the deviating users in $\mathbf{x}_{\epsilon}$ is the same as in the previous state $\mathbf{x}^{*}$, but the utility of the legitimate users (those insisting in their initial decisions) is higher in $\mathbf{x}_{\epsilon}$ than in $\mathbf{x}^{*}$. In both cases, the deviating users worsen their obtained utility. The stationary points derived above satisfy these conditions and hence they are ESS.

In detail, assume that the system has a stationary point $\mathbf{x}^{*}=(x_{1}^{*},x_{2}^{*},\ldots,x_{I}^{*},x_{0}^{*})\in X_{B}$, with $x_{0}^{*}=0$. Suppose that a user who is associated with operator $i\in\mathcal{I}$ deviates and selects another operator $j\in\mathcal{I}$. In this case, the population of users in operator $i$ decreases, $x_{i}^{\epsilon}<x_{i}^{*}$ and the population of users in operator $j$ increases, $x_{j}^{\epsilon}>x_{j}^{*}$. Initially, these two operators offered identical utility, $U_{i}(W_i,x_{i}^{*},\lambda_i)=U_{j}(W_j,x_{j}^{*},\lambda_j)$ but after the decision of the deviating user it becomes $U_{i}(W_i,x_{i}^{\epsilon},\lambda_i)>U_{j}(W_j,x_{j}^{\epsilon},\lambda_j)$. Clearly, the deviating user obtains less utility and hence there is no incentive to deviate. Similarly, if a user deviates and selects the neutral operator, he will receive reduced utility since when $\mathbf{x}^{*}\in X_B$, it is $U_{i}^{*}>U_{0},\,\forall i\in\mathcal{I}$.

Assume now that the system attains a stationary point $\mathbf{x}^{*}\in X_A$. Similarly to the previous analysis, it is straightforward that a user who deviates from $\mathbf{x}^{*}$ and moves from an operator $i\in\mathcal{I}$ to another operator $j\in\mathcal{I}$ will decrease his utility. If the user migrates to the neutral operator, his utility will not be reduced because $U_{0}$ is constant (non-congestible). However, in this case, the users that will insist in their initial choice of operator $i$ will now receive higher utility due to the move of the deviator. Due to the ESS definition and specifically according to Smith's second condition, \cite{smith-book}, this is not a preferable choice for the deviator and hence $\mathbf{x}^{*}\in X_B$ is an ESS.

Finally, when $\mathbf{x}^{*}\in X_C$, user deviation from a market operator $i\in\mathcal{I}$ to another market operator $j\in\mathcal{I}$ or to $P_0$ is not beneficial for the deviator, either because it decreases his utility or because it increases the utility of other users. In conclusion, the stationary points of the proposed revision protocol are ESS equilibriums and hence locally stable.

\section{Price Competition Among Operators} \label{Omersec:section4 - operators_competition}

In the previous section we analyzed the stationary points of
users interaction and showed that they depend on the prices selected
by the operators. Each operator anticipates the users strategy and
chooses accordingly for each epoch $\mathcal{T}$ the price that maximizes his revenue. This gives
rise to a non-cooperative price competition game $\mathcal{G_P}$ among the operators that is played in the beginning of each time epoch $\mathcal{T}$. We assume that operators are aware of the parameters of the users market and also know the values of parameters $\alpha_i,\,i\in\mathcal{I}$ and $U_0$. Specifically, we model the operators competition as a static simultaneous move normal form game of complete information, following the Bertrand competition model \cite{courcoubetis}. We are interested not only in finding the Nash equilibriums (NE) of this game but also to understand if and how the game converges to them.

We prove that $\mathcal{G_P}$ is a potential game and hence if it is played in many rounds and operators choose their prices based on the previous prices of the other operators, the game converges to a NE. In other words, we analyze the dynamics induced by the repeated play of the same game assuming that operators follow simple myopic rules.  We show that the equilibrium of the competition game depends on vector $\mathbf{\alpha}$ and the value of $U_0$. For certain combinations of these parameters, the game admits a unique equilibrium while for other combinations, it reaches one of the infinitely many equilibriums depending on the initial prices.

\subsection{Price Competition Game $\mathcal{G_P}$}
Before analyzing this game, it is important to emphasize that the revenue function depends on the price
vector $\mathbf{\lambda}$. In particular, using
equation (\ref{omereq:Operator_Revenue}), we can calculate the
revenue of operator $i$ when $\mathbf{\lambda}\in\Lambda_A$, when
$\mathbf{\lambda}\in\Lambda_B$, and when $\mathbf{\lambda}\in\Lambda_C$, denoted as $R_{i}^{A}(\cdot)$, $R_{i}^{B}(\cdot)$ and $R_{i}^{C}(\cdot)$ respectively:
\begin{equation}
R_{i}^{A}(\lambda_i)=R_{i}^{C}(\lambda_i)=\alpha_i\lambda_iNe^{-\lambda_i}
\end{equation}
and
\begin{equation}
R_{i}^{B}(\lambda_i,\lambda_{-i})=\frac{\alpha_i\lambda_iN}{e^{\lambda_i}\sum_{i=1}^{I}\alpha_ie^{-\lambda_i}}
\end{equation}
$R_{i}^{A}(\cdot)$ and $R_{i}^{C}(\cdot)$ depend only on the price
selected by operator $i$, while $R_{i}^{B}(\cdot)$ depends on the
entire price vector $\mathbf{\lambda}$. However, in all cases, the
price set ($\Lambda_A$, $\Lambda_B$ or $\Lambda_C$) to which the
price vector $\mathbf{\lambda}=(\lambda_i,\lambda_{-i})$ belongs, is
determined jointly by all the $I$ operators.

Let us now define the non-cooperative \textbf{Pricing Game} among the $I$ operators,
$\mathcal{G_P}=(\mathcal{I},\{\lambda_i\},\{R_i\})$:\\%
\begin{itemize}
\item The set of \emph{Players} is the set of the $I$ operators $\mathcal{I}=(1,2,\ldots,I)$.
\item The \emph{strategy space} of each player $i$ is its price
$\lambda_i\in[0,\lambda_{max}]$, $\lambda_{max}\in\mathcal{R}^{+}$,
and the strategy profile is the price vector
$\mathbf{\lambda}=(\lambda_1,\lambda_2,\ldots,\lambda_I)$ of the
operators.
\item The payoff function of each player is his
revenue $R_i: (\lambda_i,\mathbf{\lambda}_{-i})\rightarrow
\mathcal{R}$, where $R_i=R_{i}^{A}$ or $R_{i}^{B}$ or $R_{i}^{C}$.
\end{itemize}
The particular characteristic of this game is that each operator has
$2$ different payoff functions depending on the price profile.
Despite this characteristic, the payoff function is continuous
and quasi-concave as we prove in the Appendix, Section \ref{appendixproof:PropertOfTheRevenFunction}.
In the sequel, we analyze the best response of each operator which constitutes a reaction curve to the prices set by the other operators. The equilibrium of the game $\mathcal{G_P}$ is the intersection of the reaction curves of the operators.

\subsection{Best Response Strategy of Operators}
The best response of each operator $i$, $\lambda_{i}^{*}$, to the
prices selected by the other $I-1$ operators,
$\mathbf{\lambda}_{-i}$, depends on the users market stationary
point. Notice that for certain $\mathbf{\lambda}_{-i}$, operator $i$
may be able to select a price such that $(\lambda_i,\lambda_{-i})$
belongs to any price set ($\Lambda_A$, $\Lambda_B$ or $\Lambda_C$)
while for some $\mathbf{\lambda}_{-i}$ the operator choice will be
restricted in two or even a single price set.

\textbf{Best Response when $\mathbf{\lambda}\in\Lambda_A$:} If the
$I-1$ operators $j\in\mathcal{I}\setminus{i}$ select such prices,
$\mathbf{\lambda}_{-i}$, that the market stationary point is $\mathbf{x}^{*}\in X_A$, then operator $i$ finds the price $\lambda_{i}^{*}$ that maximizes his revenue $R_{i}^{A}(\cdot)$ by
solving the following constrained optimization problem
($\mathbf{P}_{i}^{A}$):
\begin{equation}
\max_{\lambda_i\geq 0}\alpha_i\lambda_iNe^{-\lambda_i}
\end{equation}
s.t.
\begin{equation}
\sum_{j=1}^{I}\alpha_je^{-\lambda_j}<1\label{omereq:priceconstraintA}
\end{equation}
The objective function of this problem is quasi-concave, \cite{boyd}. However, the constraint defines an open set and hence uniqueness of optimal solution is not ensured. To overcome this obstacle we substitute constraint eq. (\ref{omereq:priceconstraintA}) with the closed set:
\begin{equation}
\lambda_{i}\geq \log\frac{\alpha_i}{1-\sum_{j\neq i}\alpha_je^{-\lambda_j}} +\epsilon
\end{equation}
where $\epsilon>0$ is an arbitrary small constant number. This inequality stems from eq. (\ref{omereq:priceconstraintA}) by solving for $\lambda_i$ and adding $\epsilon$. It does not affect the problem definition and formulation nor the obtained results since, as we will prove in the sequel, operators do not select a price in the lower bound of the constraint set. After this transformation the problem has a unique optimal solution which is equal to the solution of the respective unconstrained problem, $\lambda_{i}^{*}=1$, if $(1,\lambda_{-i})\in\Lambda_A$.

\textbf{Best Response when $\mathbf{\lambda}\in\Lambda_B$:}
Similarly, when $\mathbf{\lambda}_{-i}$ is such that operator $i$
can select a price $\lambda_{i}^{*}$  with
$(\lambda_{i}^{*},\mathbf{\lambda}_{-i})\in\Lambda_B$, then his
revenue is given by the function $R_{i}^{B}(\cdot)$ and is maximized
by the solution of problem ($\mathbf{P}_{i}^{B}$):
\begin{equation}
\max_{\lambda_i\geq 0} \frac{\lambda_i\alpha_i
N}{e^{\lambda_i}\sum_{j
\in\mathcal{I}}\alpha_je^{-\lambda_j}}
\end{equation}
s.t.
\begin{equation}
\sum_{j=1}^{I}\alpha_je^{-\lambda_j}>1 \label{omereq:const-P2B-1}
\end{equation}
This is also a concave problem which would have a unique solution if the constraint set was closed and compact. Again, we substitute
the constraint with the (almost) equivalent inequality:
\begin{equation}
\lambda_{i}\leq \log\frac{\alpha_i}{1-\sum_{j\neq
i}\alpha_je^{-\lambda_j}}-\epsilon
\end{equation}
Now, the problem has a unique solution which coincides with the solution of the respective unconstrained problem, denoted $\mu_{i}^{*}$, if
$(\mu_{i}^{*},\lambda_{-i})\in\Lambda_B$ as we explain in detail in Section \ref{appendixproof:BRPricing}.

\textbf{Best Response when $\mathbf{\lambda}\in\Lambda_C$:} In this
special case, the price of each operator $i$ is directly determined
by the prices that the other operators have selected. Namely, given the vector $\mathbf{\lambda}_{-i}$, each operator $i$ has only one feasible solution (otherwise $\mathbf{\lambda}$ does not belong to $\Lambda_C$):
\begin{equation}
\lambda_{i}^{*}=\log\frac{\alpha_i}{1-\sum_{j\neq
i}\alpha_je^{-\lambda_j}} \label{omereq:priceareaC}
\end{equation}
Whether each operator $i$ will agree and adopt this price or not,
depends on the respective accrued revenue $R_{i}^{C}(\lambda_{i}^{*},\lambda_{-i})$.

We can summarize the best response price strategy of each operator
$i\in\mathcal{I}$, by defining his revenue function as follows:
\begin{equation}
R_{i}(\lambda_i,\mathbf{\lambda}_{-i},\mathbf{\alpha})=
\begin{cases}
\frac{\alpha_i\lambda_iN}{\sum_{j=1}^{I}\alpha_je^{\lambda_i-\lambda_j}} & \text{if $\lambda_i<l_0$},\\
\alpha_i\lambda_iNe^{-\lambda_i} & \text{if $\lambda_i\geq
l_0$}.
\end{cases}
\end{equation}
where $l_0=\log(\alpha_i/(1-\sum_{j\neq i}\alpha_je^{-\lambda_j}))$.
Clearly, the optimal price $\lambda_{i}^{*}$ depends both on the
prices of the other operators $\mathbf{\lambda}_{-i}$ and on
parameters $\alpha_i$, $i=1,2,\ldots,I$:
\begin{equation}
\lambda_{i}^{*}=arg\max_{\lambda_i}{R_{i}(\lambda_i,\mathbf{\lambda}_{-i},\mathbf{\alpha})}
\label{omereq:bestresponse}
\end{equation}
where $\mathbf{\alpha}=(\alpha_1,\alpha_2,\ldots,\alpha_I)$. Clearly, each operator needs to know the vector $\mathbf{\alpha}$ and to be able to observe the other operators prices in order to calculate his best response.

For each possible price vector $\mathbf{\lambda}_{-i}$ of the $\mathcal{I}\setminus{i}$ operators,
operator $i$ will solve all the above optimization problems and find
the solution that yields the highest revenue. In Section \ref{appendixproof:BRPricing} we
prove that this results in the following best response strategy:
\begin{equation}
\lambda_{i}^{*}(\mathbf{\lambda}_{-i},\mathbf{\alpha})=
\begin{cases}
1 & \text{if $(1,\mathbf{\lambda}_{-i})\in\Lambda_A$},\\
\mu_{i}^{*} & \text{if
$(\mu_{i}^{*},\mathbf{\lambda}_{-i})\in\Lambda_B$},\\
l_0 & \text{otherwise}.
\end{cases}\label{omereq:bestresponse-cases}
\end{equation}
These options are mutually exclusive. Moreover, if $\sum_{j\neq i}\alpha_j/e^{\lambda_j} \geq 1$,
the only feasible response is $\lambda_{i}^{*}=\mu_{i}^{*}$. The
dependence of $\lambda_{i}^{*}$ on parameters
$\alpha_i=W_i/(Ne^{U_0})$, $i=1,2,\ldots,I$, has interesting implications and brings into the fore the role of the regulator. Finally, observe that the transformation of the constraint set of problems $\mathbf{P}_{i}^{A}$ and $\mathbf{P}_{i}^{B}$ did not affect the best response strategy of operator $i$ since he only selects the solution of the respective unconstrained problems.

\subsection{Equilibrium Analysis of $\mathcal{G_P}$ }
The price competition game $\mathcal{G_P}$ is a finite ordinal
potential game and therefore not only has pure Nash equilibria but
also the players can reach them under any best response strategy. That is, if we consider that $\mathcal{G_P}$ is played repeatedly by the operators who update their strategy with a myopic best response method, we can show that the convergence to the equilibriums is ensured under any finite improvement path (FIP), \cite{shapley}. The potential function is:
\begin{equation}
\mathcal{P}(\mathbf{\lambda})=
\begin{cases}
\sum_{j=1}^{I}[\log{\lambda_j}-\lambda_j],\,\,\, \text{if} \sum_{j=1}^{I}\alpha_je^{-\lambda_j}\leq1,\\
\sum_{j=1}^{I}[\log{\lambda_j}-\lambda_j]-\log{(\sum_{j=1}^{I}\alpha_je^{-\lambda_j})},\,\text{else.}
\end{cases} \nonumber
\end{equation}
The detailed proof is given in Section \ref{appendixproof:ExistAndConverg}. In order to find the NE
we solve the system of equations (\ref{omereq:bestresponse-cases}),
$i=1,2,\ldots,I$ and specifically we use the iterated dominance method (Section \ref{appendixproof:DetailedAnalysisOfEquil}).

The outcome of the game $\mathcal{G_U}$ affects the strategy of operators and therefore the outcome of the game $\mathcal{G_P}$. A
price vector $(\lambda_{i}^{*},\mathbf{\lambda}_{-i}^{*})$ is an
equilibrium of the game $\mathcal{G_P}$, parameterized by the vector
$\mathbf{\alpha}=(\alpha_1,\alpha_2,\ldots,\alpha_I)$, if it satisfies:
\begin{equation}
R_i(\lambda_{i}^{*},\mathbf{\lambda}_{-i}^{*},\mathbf{\alpha})\geq
R_i(\lambda_{i},\mathbf{\lambda}_{-i}^{*},\mathbf{\alpha}),
\forall\, i\in\mathcal{I},\,\, \forall \lambda_{i}\geq
0,\,\forall\,\mathbf{x}^{*}\nonumber
\end{equation}
In order to simplify our study and focus on the results and
implications of our analysis, we assume that all operators have the
same amount of available spectrum $W_i=W$ and therefore it is also $\alpha_i=\alpha$, $\forall i\in\mathcal{I}$.

The equilibrium of the price competition game and subsequently the market stationary point $\mathbf{x}^{*}$, depend
on the value of $\alpha$. These results are summarized in Table
\ref{omertable:equilibrium} and stem from the following Theorem:
\begin{thm}
The non-cooperative game $\mathcal{G_P}$ where operators select their strategy in order to maximize their revenue, converges to one of the following pure Nash equilibria:
\begin{itemize}
\item If $\alpha\in A_1=(0,e/I)$, there is a unique Nash Equilibrium
$\mathbf{\lambda}^{*}\in\Lambda_A$, with
$\mathbf{\lambda}^{*}=(\lambda_{i}^{*}=1:\,i=1,2,\ldots,I)$ and
the respective unique market stationary point is $\mathbf{x}^{*}\in X_A$.
\item If $\alpha\in A_3=(e^{\frac{I}{I-1}}/I,\infty)$, there is a unique Nash
Equilibrium $\mathbf{\lambda}^{*}\in\Lambda_B$, with
$\mathbf{\lambda}^{*}=(\lambda_{i}^{*}=\frac{I}{I-1}:\,i=1,2,\ldots,I)$, which induces a unique respective market stationary point
$\mathbf{x}^{*}\in X_B$.
\item If $\alpha\in A_2=[e/I, e^{\frac{I}{I-1}}/I]$, there exist infinitely many
equilibria, $\mathbf{\lambda}^{*}\in\Lambda_C$, and each one of
them yields a respective market stationary point $\mathbf{x}^{*}\in
X_C$.
\end{itemize}
\end{thm}
\begin{proof}
In Section \ref{appendixproof:ExistAndConverg} of the Appendix we provide the detailed proof according to which
$\mathcal{G_P}$ is a potential game and in Section \ref{appendixproof:DetailedAnalysisOfEquil} we use iterated strict
dominance to find the Nash equilibrium $\mathbf{\lambda}^{*}$ which
depends on parameter $\alpha$.
\end{proof}

\begin{table}[t]
\caption[Pricing Game Equilibriums]{Equilibriums of $I$ operators competition for different
values of $\alpha$.} % title of Table
\centering % used for centering table
\begin{tabular}{c c c c} % centered columns (4 columns)
\hline\hline %inserts double horizontal lines
Prices/Rev. & $\alpha\in A_1$ & $\alpha\in A_2$ & $\alpha\in A_3$  \\ [0.5ex] % inserts table
%heading
\hline % inserts single horizontal line
$\lambda_{i}^{*}$ & $1$ & $\lambda_i\neq\lambda_j$                & $\frac{I}{I-1}$  \\ %
            &     &  or $\lambda_i=\lambda_j=\log{I\alpha}$ &                  \\ %
$R_{i}^{*}$       & $\frac{\alpha N}{e}$ & $R_i\neq R_j$                          & $\frac{N}{I-1}$   \\ %
            &                      &  or $R_i=R_j=\frac{N}{I}\log{I\alpha}$ &                         \\ %
$\mathbf{x}^{*}$    & $X_A$          & $X_C$                                  & $X_B$                   \\ %
\hline %inserts single line
\end{tabular}
\label{omertable:equilibrium} % is used to refer this table in the text
\end{table}

In conclusion, $\mathcal{G_P}$ is a non-cooperative game of complete information that attains certain pure Nash equilibriums (NE) which depend on parameters $\alpha_i,\,i=1,2,\ldots,I$. It is proved to be a potential game and hence the equilibriums can be reached if $\mathcal{G_P}$ is played repeatedly and operators update their strategy by simple best response or other similar utility improvement methods. If $\alpha_i$ parameters are equal, i.e. $\alpha_i=\alpha,\,\forall i\in\mathcal{I}$, then the NE is unique for $\alpha\in A_1$ or $\alpha\in A_3$. For the case $\alpha\in A_2$, the reached equilibrium depends on the initial price vector.

\section{Market Outcome and Regulation}\label{Omersec:section5 - Regulation}
The outcome of the users and operators interaction can be
characterized by the following two fundamental criteria: the
efficiency of the users market and the total revenue the operators
accrue. We show that both of them depend on parameter $\alpha$ and we further explore the impact of $W$ and $U_0$ on them.
Accordingly, we analyze the problem from a mechanism design
perspective and explain how a regulator, as the municipal WiFi
provider in the introductory example, can bias the market operation
(outcome) by adjusting the value of $\alpha$. We consider different
regulation methods and discuss their implications.

\subsection{Market Outcome and Regulation Criteria}

\subsubsection{Market Efficiency} A market is efficient if the users enjoy
high utilities in the stationary point. However, in certain
scenarios, the services provided by the $P_0$ may impose an
additional cost to the system (e.g. the cost of the municipal WiFi
provider is borne by the citizens) and hence it would be preferable
to have all the users served by the $I$ operators. Therefore, we use
the following two metrics to characterize the efficiency of the
market: \textbf{(i)} the aggregate utility ($U_{agg}$) of users in
the stationary point $\mathbf{x}^{*}$, and \textbf{(ii)} the cost $J_0=x_{0}NU_0$ incurred by the neutral operator $P_0$ for serving the portion $x_0$ of the users. Both of these
metrics depend on parameter $\alpha$ and hence on system parameters $W$ and $U_0$.

In detail:
\begin{itemize}
\item When $\alpha\in A_1=(0,e/I)$, it is $\mathbf{x}^*\in X_A$,
which means that a portion of users $x_{0}^{*}>$ selects $P_0$. The latter incurs cost of $J_0=x_{0}^{*}NU_0$ units. All users receive utility of $U_0$ units and hence the aggregate utility is $U_{agg}=NU_0$.
\item On the other hand, when $\alpha\in A_2=[e/I,e^{I/(I-1)}/I]$, it is
$\mathbf{x}^*\in X_C$. In this case, all users are served by the $I$
operators with marginal utility, i,e. $U_{i}^{*}=U_0$ for $i=1,2,\ldots,I$. There is no cost for $P_0$, i.e. $J_0=0$. Again, it is $U_{agg}=NU_0$ but unlike the previous case, there is no cost for $P_0$.
\item Finally, if $\alpha\in A_3=(e^{I/(I-1)}/I,\infty)$ it is $\mathbf{x}^*\in X_B$. All users
are served by the $I$ operators, i.e. $x_{0}^{*}=0$ and $J_0=0$, and receive high
utilities $U_{i}^{*}>U_0$, $i=1,2,\ldots,I$. The welfare is higher in this case, i.e. $U_{agg}>NU_0$.
\end{itemize}
In summary, the aggregate utility of the users changes with $\alpha$ as follows:
\begin{equation}
U_{agg}=
\begin{cases}
N U_0, & \text{if $\alpha \in A_1 \cup A_2$,}\\
N(log(\frac{WI}{N})-\frac{I}{I-1}), & \text{if $\alpha \in A_3$.}
\end{cases}
\label{omereq:SW}
\end{equation}
It can be easily verified that $U_{agg}$ is a continuous function.

We have expressed $U_{agg}$ in terms of $W$ and $U_0$ in order to investigate the impact of the system parameters in the market. When $\alpha\in A_1\cup A_2$, $U_{agg}$ increases with $U_0$ and is independent of the spectrum $W$. On the contrary,
when $\alpha\in A_3$, $U_{agg}$ increases with $W$ and is independent of $U_0$. Notice that when the value of $\alpha$
changes from $A_1$ to interval $A_2$, $U_{agg}$ remains the same but the other metric of efficiency, the cost of neutral operator $J_0$, is improved:
\begin{equation}
J_{0}=
\begin{cases}
\frac{\alpha INU_0}{e}, & \text{if $\alpha \in A_1$,}\\
0, & \text{if $\alpha \in A_2 \cup A_3$.}\\
\end{cases}
\end{equation}

\subsubsection{Revenue of Operators}
When $\alpha$ lies in the interval $A_1$, the optimal prices are
$\lambda_{i}^{*}=1$, $\forall i\in\mathcal{I}$ and all the operators
accrue the same revenue $R_i^*=\alpha Ne^{-1}=We^{-(U_0+1)}$, which is proportional to
$\alpha$, increases with the available spectrum $W$, decreases with $U_0$ and is independent of
the number $N$ of users. In Figure \ref{omerfig:figure3} we depict the revenue of
each operator for different values of $\alpha$, in a duopoly market. Notice that the revenue
increases linearly with $\alpha\in(0, e/2)$.

When $\alpha\in A_2$, the competition of the operators may attain
different equilibria, $\mathcal{\lambda}^*\in \Lambda_C$, depending
on the initial prices and on the sequence the operators update their
prices. In Figure \ref{omerfig:figure1} we present the revenue of
two operators (duopoly) at the equilibrium, for various initial
prices and for $\alpha=e\in A_2$. Here we assume that the $1^{st}$
operator is able to set his price $\lambda_1(0)$ before the $2^{nd}$
operator. Also, in Figure \ref{omerfig:figure3} we illustrate the dependence of the revenue of
the operators on the value of $\alpha$ when it lies in $A_2$, given
that $\lambda_1(0)=1.1$. For certain prices, e.g. when
$\lambda_1(0)=\log2\alpha$, both operators accrue the same revenue
at the equilibrium, $R_{1}^{*}=R_{2}^{*}=\frac{N\log2\alpha}{2}$.

If $\alpha\in A_3=(e^{I/(I-1)},\infty)$ all operators set their
prices to $\lambda_i^*=I/(I-1)$ and get $R_i^*=N/(I-1)$ units, as
shown in Table \ref{omertable:equilibrium}. Figure
\ref{omerfig:figure2} depicts the competition of two operators and
the convergence to the respective Nash equilibria for $\alpha=e^3
\in A_3$. We assume that both operators have selected prices
$\lambda_{1}(0)=\lambda_{2}(0)=\log{2\alpha}\approx 3.7$. However,
this price vector does not constitute a NE and hence an operator
(e.g. the $1^{st}$) can temporarily increase his revenue by
decreasing his price to $\lambda_1=3$. Accordingly, the other
operator ($2^{nd}$) will react by reducing his price to
$\lambda_2=2.5$. Gradually, the competition of the operators will
converge to the NE where both of them will set
$\lambda_1^*=\lambda_2^*=2/(2-1)=2$ and will have revenue
$R_{1}^{*}=R_{2}^{*}=1$. Interestingly, the revenue of both
operators in the equilibrium is lower than their initial revenue
when they did not compete. Finally, notice that, unlike the aggregate utility $U_{agg}$,
the revenue of the operators depends only on $\alpha=W/(Ne^{U_0})$ and not the specific values of $W$ and $U_0$.

Before we proceed, let us summarize the above results:
\begin{itemize}
\item If $\alpha \in A_1=(0,e/I)$, it is $R_{i}^{*}=\alpha Ne^{-1}=We^{-(U_0+1)}$, $i=1,2,\ldots,I$. Operators receive equal revenue which is \textbf{(i)} proportional to $W$, \textbf{(ii)} inversely proportional to $U_0$ and \textbf{(iii)} independent of the number of users $N$.
\item If $\alpha \in A_2=[e/I,\frac{e^{I/(I-1)}}{I})$, $R_{i}^{*}$ depends on the initial prices operators select. In the particular case that a single operator $i$ sets first his price $\lambda_i$ so as to be $\lambda_i(0)=\log{I\alpha}$, then all operators obtain finally equal revenue $R_{i}^{*}=\frac{N\log{I\alpha}}{I}$.
\item If $\alpha \in A_3=[\frac{e^{I/(I-1)}}{I},\infty)$, it is $R_{i}^{*}=\frac{N}{I-1}$. Operators receive equal revenue which is \textbf{(i)} proportional to $N$, \textbf{(ii)} independent of $U_0$ and $W$.
\end{itemize}

\subsection{Regulation of the Wireless Service Market}

Since both the market efficiency and the operator revenue depend on
$\alpha$ and system parameters $W$ and $U_0$, a regulating agency can act as a \emph{mechanism designer}
and steer the outcome of the market in a more desirable equilibrium
according to his objective. This can be achieved by determining
directly or indirectly (e.g. through pricing) the amount of spectrum
$W$ each operator has at his disposal, or by intervening in the
market and setting the value $U_0$ as the example with the municipal
WiFi Internet provider. This process is depicted in Figure
\ref{omerfig:MarketRegulation2}.

\begin{figure}[t]
\begin{center}
\epsfig{figure=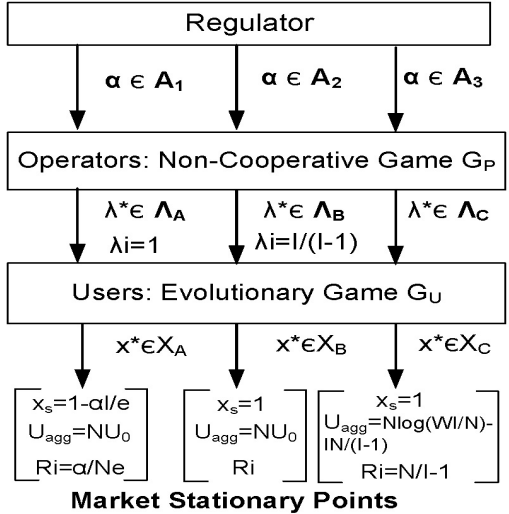,
width=7cm,height=7cm}
\end{center}
\caption[Regulation and Interdependency of Operators and Users Games
]{The regulator selects parameter $\alpha$, the operators compete
and select the respective optimal prices $\lambda_{i}^{*}$, and
then, the users are divided among the operators.}
\label{omerfig:MarketRegulation2}
\end{figure}

\subsubsection{Regulating to Increase Market Efficiency}

First, we highlight the impact of parameters $W$ and $U_0$ on the efficiency metrics. This is of crucial importance because tuning $W$ or $U_0$ has different implications for the regulator and the market. For example, as it is explained below, the regulator can achieve the same level of market efficiency either by selling more spectrum to operators, e.g. by decreasing the spectrum price, or by allocating more spectrum to the neutral operator:
\begin{itemize}
\item \emph{Assume that $U_0$ is fixed}. As the allocated spectrum $W$ to each operator increases, aggregate utility $U_{agg}$ remains constant until parameter $\alpha$ increases up to $\alpha\geq \frac{e^{I/(I-1)}}{I}$. When $\alpha \in A_3$, $U_{agg}$ is log-proportional to $W$. Also, the cost $J_0$ increases with $W$, as long as $\alpha\in A_1$, and becomes zero for larger values of $\alpha$.

\item \emph{Assume that $W$ is fixed}. $U_{agg}$ increases with $U_0$ as long as $\alpha\in A_1 \cup A_2$. For larger values of $\alpha$, $U_{agg}$ does not depend directly on $U_0$. Additionally, the cost $J_0$ increases with $U_0$ as long as $\alpha\in A_1$ while for larger values of $\alpha$ it becomes zero.
\end{itemize}

Let us now give a specific scenario for regulation. Assume that initially $\alpha \in A_1=(0,e/I)$. Hence, a portion of users is not served by anyone of the $I$ operators,
$x_{0}^{*}>1$ and all the users receive utility equal to $U_0$. The
regulator can improve the market efficiency, i.e. increase $U_{agg}$ and decrease $J_0$, by increasing the value of $\alpha$. This
can be achieved either by increasing $W$ or decreasing $U_0$. Let us
assume that the regulator selects the first method. For example, he
can change the price of $W$ and allow the operators to acquire more
spectrum. If $W$ is increased until $\alpha=e/I$, then the market
stationary point $\mathbf{x}^{*}$ switches to $X_B$. In this case,
all users are served by the market, $x_{0}^{*}=0$, but they still
receive only marginal utility, $U_{agg}=N U_0$. If the regulator
provides even more spectrum $W$ to operators so as $\alpha>e^{I/(I-1)}/I$, then $x_{0}^{*}=0$ and moreover the users perceive higher utility because $U_{agg}$ increases proportional
to $\log{W}$, eq. (\ref{omereq:SW}). Obviously, the improvement in market efficiency comes at the cost (\emph{opportunity cost}) of the additional spectrum the regulator must provide to operators.

On the other hand, the regulator may prefer to directly intervene in
the market through $P_0$ and tune $U_0$. If $U_0$ decreases, the value of $x_{0}^{*}$ decreases and users
return to the market (to the $I$ operators). The portion of users $x_{0}^{*}$ becomes zero when $\alpha=e/I$. This way, the cost of the regulator $J_0$ decreases (since $P_0$ serves less
users) but at the same time the aggregate utility, $U_{agg}=NU_0$, is also reduced.
Namely, $U_{agg}$ decreases linearly with $U_0$ until
$\alpha=e^{I/(I-1)}/I$ and remains constant for larger values of
$U_0$, eq. (\ref{omereq:SW}). Again, the decision of the regulator
depends on his cost and on the efficiency he wants to achieve.
In conclusion, depending on they system parameters ($N,\,W,\,I$) the efficiency of the market may be improved either by increasing the resources of operators (sell more spectrum) or by rendering highly competitive the services provided by the neutral operator $P_0$.

\begin{figure}
\center
%%%\begin{minipage}{8.3cm}
\epsfig{figure=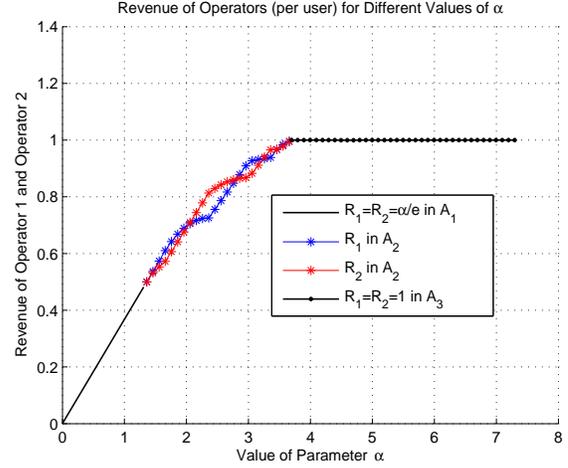, width = 8.5cm,
height=6.5cm} \caption[Operators Revenue for different values of
$\alpha$]{The outcome of the operator competition ($\mathcal{G_P}$
equilibrium) for different values of parameter $\alpha$, i.e. in
different intervals.} \label{omerfig:figure3}
%%\end{minipage}
%\begin{minipage}{8.3cm}
\end{figure}%--
\begin{figure}%--
\center \epsfig{figure=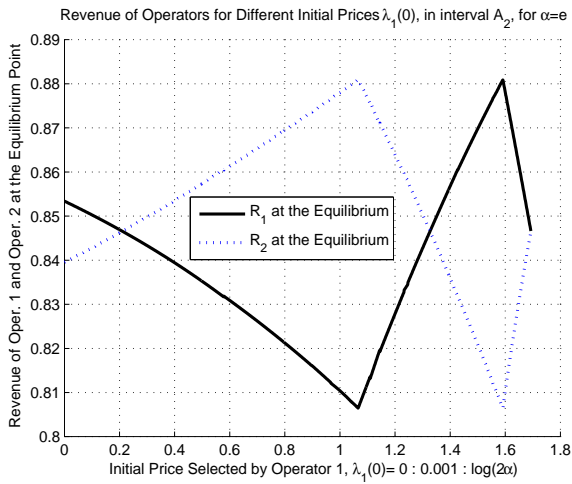, width =
8.5cm, height=6.5cm} \caption[Operators Revenue for different
initial prices]{The outcome of the competition of two
operators, with $\alpha=e \in A_2$ and $N=1000$. Operator $1$ is
assumed to set his price $\lambda_1(0)$ first. $R_{1}^{*}$ and $R_{2}^{*}$ depend on $\lambda_1(0)$.}\label{omerfig:figure1}
%%\end{minipage}
\end{figure}

\subsubsection{Regulating for Revenue}
As illustrated in Table \ref{omertable:equilibrium}, the revenue of
the operators increases proportionally to $\alpha$ for $\alpha\in
A_1$, and proportionally to $\log{\alpha}$ for $\alpha\in A_2$,
while it remains constant when $\alpha\in A_3$. Notice that the
revenue, unlike the market efficiency, depends on the value of
$\alpha$ and not on the specific combination of $W$ and $U_0$. These
results are presented in Figure \ref{omerfig:figure4} for a market
with $I=3$ operators and $N=1000$ users. In the upper plot, it is
$U_0=0.1$ and the regulator increases the value of $\alpha$ by
increasing $W$. The aggregate utility is constant and equal to
$U_{agg}=N U_0=100$ for $\alpha<e^{3/(3-1)}/3 \approx 1.5$ while it
increases proportionally to $\log{W}$ for $\alpha>1.5$. Obviously,
increasing the spectrum of operators improves both their revenue and
the efficiency of the market.

In the lower plot, the spectrum at the disposal of each operator is
constant, $W=5000$, and the regulator increases the value of
$\alpha$ by decreasing $U_0$. In this case, the total revenue
increases but at the expense of market efficiency. When $\alpha\in
A_1 \cup A_2 =(0,e^{1.5}/3]$, the aggregate utility $U_{agg}$ is
reduced as $U_0$ decreases but for $\alpha> e^{1.5}/3$ it remains
constant. Notice that for very small values of $\alpha$, $U_{agg}$ is large. However, this desirable result comes at a cost for the regulator. Namely, in this case only a small portion of users are served by the market, while the rest of them select $P_0$. Therefore, the incurred cost $J_0$ for the regulator is high.

\begin{figure}
\center
%%\begin{minipage}{8.3cm}
\epsfig{figure=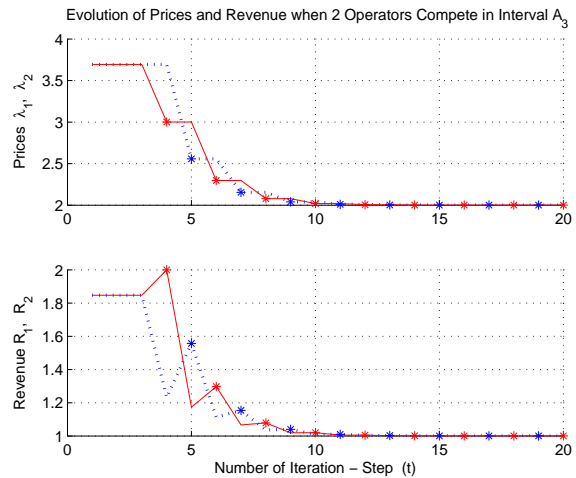, width
= 8.5cm, height=6.5cm} \caption[Evolution of Operators Price
Competition]{Evolution of operator competition for $\alpha = e^3 \in
A_3$. The game is played repeatedly and operators updated myopically their price based on the previous strategy of the other operators.} \label{omerfig:figure2}
%%\end{minipage}
\end{figure}%--
%%\begin{minipage}{8.3cm}
\begin{figure}%%--
\center \epsfig{figure=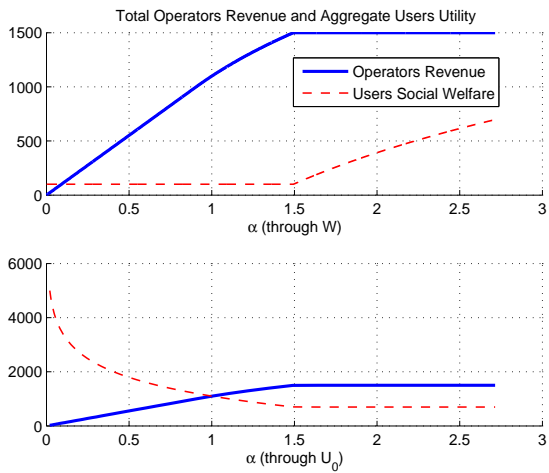,
width = 8.5cm, height=6.5cm} \caption[Market Revenue and Aggregate
Welfare for various $\alpha$]{Total revenue of operators and
aggregate utility of the users for different values of parameter
$\alpha$. In the upper plot, the value of $\alpha$ changes through
$W$ while in the lower plot it changes through the tuning of $U_0$.}
\label{omerfig:figure4}
%%\end{minipage}
\end{figure}

Another interesting point in Figure \ref{omerfig:figure4} is the following. In the upper subplot, for $U_0=0.1$, the total operators revenue is $R_{tot}=1500$ units and the aggregate utility is $U_{agg}=100$, achieved by increasing the spectrum $W$ until $W=1657.8$ units, which yields $\alpha=1.5$. In the lower plot, the same total revenue is reached for $W=5000$, and decreasing $U_0$ until $U_0=1.204$ units. In this case, the aggregate utility is $U_{agg}=1204$ units. If, for example, the regulator is interesting only in maximizing the revenue of operators, then he would prefer the first method since it requires less spectrum and lower value for $U_0$.

\section{Conclusions} \label{Omersec:section6 - conclusions}

In this paper, we studied the operators price competition in a wireless services market where users have a certain reservation utility $U_0$. We modeled the users interaction as an evolutionary game and the competition of the operators as a non cooperative game of complete information. We proved that the latter is a potential game and hence has pure Nash equilibriums. The two games are realized in different time scale but they are interrelated. Additionally, both of them depend on the reservation utility $U_0$ and the amount of spectrum $W$ each operator has at his disposal. Accordingly, we considered a regulating agency and discussed how he can intervene and change the outcome of the market by tuning either $U_0$ or $W$. Various regulation methods yield different market outcomes and induce a different cost for the regulator.

%%----\begin{subappendices}
\appendix

\section{Analysis of the Evolutionary Game $\mathcal{G_U}$}

\subsection{Derivation of Evolutionary Dynamics} \label{appendixproof:derivofevolutionarydynamics}
Here, we derive the new differential equations that describe the
evolution of the market of the users under the new introduced
revision protocol. Recall that, the latter is described by the following
equations:
\begin{equation}
\rho_{ij}(t)=x_j(t)[U_j(t)-U_i(t)]_{+}, \forall i,j\in\mathcal{I}
\label{eq:Revision_Oper}
\end{equation}
\begin{equation}
\rho_{i0}(t)=\gamma[U_0-U_i(t)]_{+}, \forall i\in\mathcal{I}
\label{eq:Revision_toNeut}
\end{equation}
\begin{equation}
\rho_{0i}(t)=x_i(t)[U_i(t)-U_0]_{+}, \forall i\in\mathcal{I}
\label{eq:Revision_fromNeut}
\end{equation}
where $\rho_{ij}(t)$ is the rate at which users associated with operator $i$ switch to operator $j$ in time slot $t$, $\rho_{i0}(t)$ is the switch rate from operator $i$ to neutral operator $P_0$ and $\rho_{0i}(t)$ the rate at which users return from $P_0$ to an operator $i\in\mathcal{I}$ in the market. The constant value
$\gamma\in R^{+}$ represents the frequency of the direct selection.

For imitation-based revision protocols, the dynamics of the system can be described with the well-known
replicator dynamics \cite{sandholm1}. The hybrid revision protocol defined in equations (\ref{eq:Revision_Oper}),
(\ref{eq:Revision_toNeut}) and (\ref{eq:Revision_fromNeut}) is in part imitation-based ($\rho_{ij}(t)$ and $\rho_{0i}(t)$) and in part a
probabilistic direct selection of the neutral operator ($\rho_{i0}(t)$). Therefore, the respective evolutionary dynamics of
the system cannot be described by the replicator dynamic equations which correspond to the pure imitation mechanism. We have to stress
that the hybrid protocol that we introduce, differs from the hybrid protocol in \cite{sandholm1} in that users select directly only the neutral operator and not the other $I$ operators.

The portion of users $x_i$ who are associated with operator $i$
changes from time $t$ to the time $t+\delta t$, according to the
following equation:
\begin{eqnarray}
x_i(t+\delta t)&=&x_i(t)-x_i(t)\delta t\sum_{j\neq
0}x_j(t)(U_j(t)-U_i(t))_{+}  \nonumber \\
&-& x_i(t)\delta t\gamma(U_0(t)-U_i(t))_{+} \nonumber \\
&+& \sum_{j=0}^{I}\delta tx_j(t)x_i(t)(U_i(t)-U_j(t))_{+}
\end{eqnarray}
for $\delta t\rightarrow 0$ we obtain the derivative:
\begin{eqnarray}
\frac{d x_i(t)}{dt}&=& x_i(t)[\sum_{j\neq
0}x_j(t)U_i(t)-\sum_{j\neq 0}x_j(t)U_j(t) \nonumber \\
&-& \gamma(U_0-U_i)_{+} + x_0(t)(U_i-U_0)_{+}] \nonumber
\end{eqnarray}
or, if we omit the time index and rewrite the equation:
\begin{eqnarray}
\frac{d x_i(t)}{dt}&=& x_i[U_i-U_{avg}-x_0(U_i-U_0) - \gamma
(U_0-U_i)_{+} \nonumber \\
&+&x_0(U_i-U_0)_{+}]\nonumber
\end{eqnarray}
which can be analyzed in:
\begin{equation}
\frac{d x_i(t)}{dt}=x_i (U_i-U_{avg}),\,\,\forall i \in
\mathcal{I}^+ \label{eq:x-i-dot}
\end{equation}
\begin{equation}
\frac{d x_j(t)}{dt}=x_j [U_j-U_{avg}-(\gamma -
x_0)(U_0-U_j)],\,\,\forall j \in \mathcal{I}^- \label{eq:x-j-dot}
\end{equation}
where $\mathcal{I^+}$ is the set of operators offering utility
$U_i(t)\geq U_0$, and $\mathcal{I^-}$ is the set of operators
offering utility $U_j(t)<U_0$.

The dynamics of the population $x_0$ can be derived in a similar
way:
\begin{eqnarray}
x_0(t&+&\delta t)=x_0(t)-x_0(t)\delta t\sum_{i\neq 0}x_i(t)
(U_i-U_0)_{+} \nonumber \\ 
&+& \sum_{i \neq 0}x_i(t)\delta t\gamma(U_0-U_i)_{+}
\end{eqnarray}
which can be written as:
\begin{equation}
\frac{d x_0(t)}{dt}=(x_0 \sum_{i \in I^+}{x_i(U_0-U_i)} +
\gamma\sum_{j \in \mathcal{I}^-}{x_j(U_0-U_j)}) \label{eq:x-0-dot}
\end{equation}

Equations (\ref{eq:x-i-dot}), (\ref{eq:x-j-dot}) and
(\ref{eq:x-0-dot}) describe the evolutionary dynamics of game
$\mathcal{G_U}$.

\subsection{Analysis of Stationary Points}\label{appendixproof:AnalysisOfStationaryPoints}
Despite the different dynamics, the system reaches the same
stationary points as if users where employing the typical imitation revision protocol. In detail, the market state vector at a fixed
point, $\mathbf{x}^{*}=(x_{i}^{*}$, $x_{j}^{*}$, $x_{0}^{*}$: $\forall\, i \in
\mathcal{I}^+$, $\forall\,j \in \mathcal{I}^-)$, can be found by the
following set of equations:
\begin{equation}
\frac{d x_i(t)}{dt}=\frac{d x_j(t)}{dt}=\frac{d
x_0(t)}{dt}=0\,\,\forall i\in\mathcal{I}^{+},\,j\in\mathcal{I}^{-} \label{eq:derivatives-are-zero}
\end{equation}
\begin{lem}
The stationary points of the evolutionary dynamics defined in
equations (\ref{eq:x-i-dot}), (\ref{eq:x-j-dot}) and
(\ref{eq:x-0-dot}) are identical to the stationary points of the
ordinary replicator dynamics \cite{sandholm1} given by:
\begin{equation}
\Dot{x}_i(t)=0 \Rightarrow
x_i(t)[U_i(t)-U_{avg}(t)]=0,\,\forall\,i\in\mathcal{I}
\label{eq:Pi_stationarypoints}
\end{equation}
and
\begin{equation}
\Dot{x}_0(t)=0\Rightarrow x_0(t)[U_0-U_{avg}(t)]=0
\label{eq:P0_stationarypoints}
\end{equation}
\label{lemma0}
\end{lem}
\begin{proof}
First we prove that, in any stationary point, $x_{j}^{*},\,j\in\mathcal{I}^{-}$ should be
equal to zero. We prove this claim by contradiction. Assume that $x_{j}^{*}>0$. Since $U_{avg}\geq U_0>U_j$, this implies
that there should be at least one operator $i$ with $U_i>U_{avg}$
and $x_{i}^{*}>0$. Therefore ($U_i-U_{avg}$) cannot be equal to zero
$\forall i \in \mathcal{I}^+$, and $\dot{x}_{i}$ will be nonzero for at least one operator. Therefore
(\ref{eq:derivatives-are-zero}) cannot be satisfied, if $x_{j}^{*}\neq 0$.

When $x_{j}=0$, the evolutionary dynamics given by eq.
(\ref{eq:x-i-dot}), (\ref{eq:x-j-dot}) and (\ref{eq:x-0-dot}) reduce
to ordinary replicator dynamics:
\begin{equation}
\Dot{x}_i(t)=x_i(t)[U_i(t)-U_{avg}(t)]\ \forall i\in\mathcal{I}\nonumber
\end{equation}
\begin{equation}
\Dot{x}_0(t)=x_0(t)[U_0-U_{avg}(t)] \label{eq:Replicator_dynamics}
\end{equation}
Stationary points are identical to the stationary points of the typical replicator dynamics, \cite{sandholm1}.
\end{proof}
Due to this lemma, the stationary points for the users
population associated with each operator $i\in\mathcal{I}$ should
satisfy one of the following conditions: $(i)$ $x_{i}^{*}=0$, or
$(ii)$ $x_{i}^{*}>0$ and $U_{i}^{*}=U_{avg}$. Similarly, for the
neutral operator $P_0$, eq. (\ref{eq:P0_stationarypoints}), it must
hold: $(i)$ $x_{0}^{*}=0$ and $U_{0}<U_{avg}$, $(ii)$
$x_{0}^{*}>0$ and $U_{0}=U_{avg}$ or $(iii)$ $x_{0}^{*}=0$ and
$U_{0}=U_{avg}$. The case $x_{i}^{*}=0$ implies zero revenue for
the $i^{th}$ operator and hence case $(i)$ does not constitute a
valid choice. Therefore, there exist in total $3$ possible
combinations (cases) that will satisfy the stationarity properties
given by eq. (\ref{eq:Pi_stationarypoints}) and
(\ref{eq:P0_stationarypoints}):
\begin{itemize}
\item \textbf{Case A}: $x_{i}^{*}$, $x_{0}^{*}>0$ and
$U_{i}^{*}=U_0$, $i\in\mathcal{I}$.
\item \textbf{Case B}: $x_{i}^{*}$, $x_{j}^{*}>0$, $x_{0}^{*}=0$ and
$U_{i}^{*}=U_{j}^{*}$, with $\,$
$U_{i}^{*},\,U_{j}^{*}>U_0$,$\,\forall\,i,\,j\in\mathcal{I}$.
\item \textbf{Case C}: $x_{i}^{*}$, $x_{j}^{*}>0$, $x_{0}^{*}=0$ and
$U_{i}^{*}=U_{j}^{*}=U_0$, $\,\forall\,i,\,j\in\mathcal{I}$.
\end{itemize}
We find now the exact value of the market state vector at the
equilibrium (stationary point) $\mathbf{x}^{*}$ for each case.
First, we define for every operator $i\in\mathcal{I}$ the scalar
parameter $\alpha_i=W_i/(Ne^{U_0})$ and the respective vector
$\mathbf{\alpha}=(\alpha_i\,:\,i\in\mathcal{I})$.

We can find the stationary points for \textbf{Case A} by using the
equation $U_{i}(W_i,x_{i}^{*},\lambda_i)=U_0$ and imposing the
constraint $x_{0}^{*}>0$:
\begin{eqnarray}
U_{i}(W_i,x_{i}^{*},\lambda_i)&=&\log{\frac{W_i}{Nx_i^{*}}}-\lambda_i=U_0
\\ \nonumber
&\Rightarrow& x_{i}^{*}=\frac{W_i}{Ne^{\lambda_i+U_0}}=\alpha_ie^{-\lambda_i}
\end{eqnarray}
and
\begin{equation}
x_{0}^{*}>0 \Rightarrow 1-\sum_{i=1}^{I}\alpha_ie^{-\lambda_i}>0 \Rightarrow \sum_{i=1}^{I}\alpha_ie^{-\lambda_i}<1
\end{equation}
Apparently, the state vector $\mathbf{x}^{*}$ depends on the
operators' price vector $\mathbf{\lambda}$. Therefore, we define the
set of all possible \textbf{Case A} stationary points, $X_A$, as
follows:
\begin{equation}
X_A=\Bigg\{x_{i}^{*}=\alpha_ie^{-\lambda_i},\forall
i\in\mathcal{I},
x_{0}^{*}=1-\sum_{i=1}^{I}\alpha_ie^{-\lambda_i}:\,\mathbf{\lambda}\in
\Lambda_A \Bigg\}\nonumber
\end{equation}
where $\Lambda_A$ is the set of prices for which a stationary point
in $X_A$ is reachable, i.e. $x_{0}^{*}>0$:
\begin{equation}
\Lambda_A=\Bigg\{
(\lambda_1,\lambda_2,\ldots,\lambda_I)\,:\,\sum_{i=1}^{I}\alpha_ie^{-\lambda_i}<1
\Bigg\} \label{eq:Lambda_A}\nonumber
\end{equation}

Similarly, for \textbf{Case B}, we calculate the stationary points
by using the set of equations
$U_{i}(W_i,x_{i}^{*},\lambda_i)=U_{j}(W_j,x_{j}^{*},\lambda_j)$,
$\forall\,i,j\in\mathcal{I}$, which yields:
\begin{equation}
\log{\frac{W_1}{Nx_{1}^{*}}}-\lambda_{1}=\log{\frac{W_2}{Nx_2^{*}}}-\lambda_2=\ldots=\log{\frac{W_i}{Nx_i^{*}}}-\lambda_i
\end{equation}
or, equivalently:
\begin{equation}
x_j^{*}=x_i^{*}\frac{e^{\lambda_i}\alpha_j}{e^{\lambda_j}\alpha_i}\
\ \forall{i,j \in\mathcal{I}} \label{eq:xi-xj-relation}
\end{equation}
Moreover since $x_0^{*}=0$ for \textbf{Case B}, the following holds:
\begin{equation}
\sum_{i \in\mathcal{I}}{x_i^{*}}=1 \label{eq:area-B-constraint}
\end{equation}
Using (\ref{eq:xi-xj-relation}) and (\ref{eq:area-B-constraint}),
\begin{equation}
x_i^{*}=\frac{\alpha_i}{e^{\lambda_i}\sum_{j
\in\mathcal{I}}\alpha_je^{-\lambda_j}} \label{eq:x_i_B}
\end{equation}
Additionally, $U_i>U_0$ implies that:
\begin{equation}
\log{\frac{W_i}{Nx_i^{*}}}-\lambda_i>U_0 \Rightarrow
x_i^{*}<\alpha_ie^{-\lambda_i} \label{eq:constraint2}
\end{equation}
Using (\ref{eq:area-B-constraint}) and (\ref{eq:constraint2}),
\begin{equation}
\sum_{i=1}^{I}{x_i^{*}}<\sum_{i=1}^{I}\alpha_ie^{-\lambda_i} \Rightarrow \sum_{i=1}^{I}\alpha_ie^{-\lambda_i}>1
\label{eq:constraint3}
\end{equation}
Therefore, according to (\ref{eq:x_i_B}) and (\ref{eq:constraint3}),
we define the set of all possible \textbf{Case B} stationary points,
$X_B$, as follows:
\begin{equation}
X_B=\Bigg\{x_{i}^{*}=\frac{\alpha_i}{e^{\lambda_i}\sum_{j=1}^{I}\alpha_je^{-\lambda_j}},\forall
i\in\mathcal{I},x_{0}^{*}=0:\,\mathbf{\lambda}\in \Lambda_B \Bigg\}\nonumber
\end{equation}
where $\Lambda_B$ is the set of prices for which a stationary point
in $X_B$ is feasible, $U_{i}^{*}>U_0$:
\begin{equation}
\Lambda_B=\Bigg\{
(\lambda_1,\lambda_2,\ldots,\lambda_I)\,:\,\sum_{i=1}^{I}\alpha_ie^{-\lambda_i}>1
\Bigg\} \label{eq:Lambda_B}\nonumber
\end{equation}

Finally, the stationary points for the \textbf{Case C} solution must
satisfy the following:
\begin{equation}
U_i = U_0,\ \ \ x_0^{*}=0
\end{equation}
which yields:
\begin{equation}
x_i^{*}=\alpha_ie^{-\lambda_i},\ \ \
\sum_{i=1}^{I}\alpha_ie^{-\lambda_i}=1
\end{equation}
Therefore, we define the set of all possible \textbf{Case C}
stationary points, $X_C$, as follows:
\begin{equation}
X_C=\Bigg\{
x_{i}^{*}=\alpha_ie^{-\lambda_i},\forall\,i\in\mathcal{I},x_{0}^{*}=0:\,\mathbf{\lambda}\in
\Lambda_C \Bigg\}\nonumber
\end{equation}
with
\begin{equation}
\Lambda_C=\Bigg\{
(\lambda_1,\lambda_2,\ldots,\lambda_I)\,:\,\sum_{i=1}^{I}\alpha_ie^{-\lambda_i}=1
\Bigg\}\nonumber
\end{equation}

\section{Analysis of the Pricing Game $\mathcal{G_P}$}\label{appendixproof:AnalysisOfPricingGame}

First, we show that the revenue of each operator $i\in\mathcal{I}$
is a continuous and a quasi-concave function. Secondly, we analyze
best response pricing in game $\mathcal{G_P}$. Then, we derive the
Nash equilibriums (NEs) of the game using iterated strict dominance.
Finally, we prove convergence to these equilibriums by showing that $\mathcal{G_P}$ is a potential game.

\subsection{Properties of the Revenue Function}\label{appendixproof:PropertOfTheRevenFunction}
The revenue function of each operator $i$ is given by the following
equation:
\begin{equation}
R_{i}(\lambda_i,\mathbf{\lambda}_{-i})=
\begin{cases}
\displaystyle\frac{\alpha_i\lambda_iN}{e^{\lambda_i}\sum_{j=1}^{I}\alpha_je^{-\lambda_j}} & \text{if $\lambda_i<l_0$},\\
\displaystyle\alpha_i\lambda_iNe^{-\lambda_i} & \text{if
$\lambda_i\geq l_0$}.
\end{cases}
\label{eq:Revenue_function}
\end{equation}
where $l_0=\log(\alpha_i/(1-\sum_{j\neq i}{\alpha_je^{-\lambda_j}}))$.

Each component (for each case) is a positive function which is also
\emph{log-concave}. This means that it is a quasiconcave function
and hence uniqueness of optimal solution is ensured for a proper constraint set. Namely, it is:
\begin{equation}
f_{A}(\lambda_i)=\log\alpha_i\lambda_iNe^{-\lambda_i}=\log\alpha_i\lambda_iN-\lambda_i
\end{equation}
and
\begin{equation}
f_{A}(\lambda_i)^{(1)}=\frac{1}{\lambda_i}-1\Rightarrow
f_{A}(\lambda_i)^{(2)}=\frac{-1}{\lambda_{i}^{2}}<0
\end{equation}
Hence, $f_A(\cdot)$ which is the log-function of $R_{i}^{A}(\cdot)$, is
concave which means that the later is log-concave and since it is
$R_{i}^{A}(\lambda_i)>0$, it is also quasi-concave. Similarly, for
the other component of the revenue function:
\begin{equation}
f_{B}(\lambda_i,\lambda_{-i})=
\log{\frac{\alpha_i\lambda_iN}{\alpha_i+\beta
e^{\lambda_i}}}=\log{\alpha_i\lambda_iN}-\log{\alpha_i+\beta
e^{\lambda_i}}
\end{equation}
where $\beta=\sum_{j\neq i}\alpha_je^{-\lambda_i}$. The
second derivative is:
\begin{equation}
f_{B}(\lambda_i,\lambda_{-i})^{(2)}=
\frac{-1}{\lambda_{i}^{2}}-\frac{\alpha_i\beta
e^{\lambda_i}}{(\alpha_i+\beta e^{\lambda_i})^{2}}<0
\end{equation}
Hence, $R_{i}^{B}(\cdot)$ is also quasiconcave. Finally, it is easy
to see that the function is continuous:
\begin{equation}
R_{i}^{A}(l_0,\lambda_{-i})=R_{i}^{B}(l_0,\lambda_{-i})=N(1-\beta)\log{\frac{\alpha_i}{1-\beta}}
\end{equation}

\subsection{Best Response Pricing in $\mathcal{G_P}$} \label{appendixproof:BRPricing}
Each operator $i$ finds his best response price $\lambda_{i}^{*}$
for each price profile of the other $I-1$ operators by solving
the following optimization problems. For the case the price vector
belongs to the set $\Lambda_A$, $\mathbf{\lambda}\in\Lambda_A$,
($\mathbf{P}_{i}^{A}$):
\begin{equation}
\max_{\lambda_i\geq 0}\alpha_i\lambda_iNe^{-\lambda_i}
\end{equation}
s.t.
\begin{equation}
\sum_{j=1}^{I}\alpha_je^{-\lambda_j}<1
\end{equation}
In order to ensure the uniqueness of the problem solution, we transform the constraint set to a closed and compact set as follows:
\begin{equation}
\lambda_{i}\geq\log{\frac{\alpha_i}{1-\sum_{j\neq
i}\alpha_je^{-\lambda_j}}}+\epsilon
\end{equation}
where $\epsilon>0$ is an arbitrary small positive constant number. As we will show immediately this transformation of the constraint set does not affect the solution of the game. The problem now is quasi-concave with a closed and compact constraint set and hence it has a unique optimal
solution, \cite{boyd} which we denote $\lambda_{i}^{A}$ and it is:
\begin{equation}
\lambda_{i}^{A}=1,\,\texttt{or}\,\lambda_{i}^{A}=\log{\frac{\alpha_i}{1-\sum_{j\neq
i}\alpha_je^{-\lambda_j}}}+\epsilon
\end{equation}
The value $\lambda_{i}^{A}=1$ is the optimal solution of the respective unconstrained problem, which yields
optimal revenue $R_{i}^{A}=\alpha N/e$, and it is feasible if $\mathbf{\lambda}=(1,\lambda_{-i})\in\Lambda_A$. Otherwise, since
$R_{i}^{A}(\cdot)$ is a decreasing function of $\lambda_i$, operator $i$ can only select the minimum price $\lambda_{i}^{A}$ such that
$(\lambda_{i}^{A},\lambda_{-i})\in\Lambda_A$.

Similarly, when the price vector belongs to the set $\Lambda_B$, i.e. $\mathbf{\lambda}\in\Lambda_B$, the revenue maximization problem for
each operator $i\in\mathcal{I}$ ($\mathbf{P}_{i}^{B}$) is:
\begin{equation}
\max_{\lambda_i\geq 0} \frac{\lambda_i\alpha_i
N}{e^{\lambda_i}\sum_{j \in\mathcal{I}}\alpha_je^{-\lambda_j}}
\end{equation}
s.t.
\begin{equation}
\sum_{j
\in\mathcal{I}}\alpha_je^{-\lambda_j}>1\ \label{eq:const-P2B-1}
\end{equation}
Similarly to the previous analysis, we transform the constraint set to a closed and compact set by using the following inequality:
\begin{equation}
\lambda_{i}\leq\log{\frac{\alpha_i}{1-\sum_{j\neq
i}\alpha_je^{-\lambda_j}}}-\epsilon
\end{equation}
This is also a concave problem which has unique solution and can be
either the optimal solution of the respective unconstrained problem,
$\lambda_{i}^{*}$ if $(\lambda_{i}^{*},\lambda_{-i})\in\Lambda_B$,
or the maximum price for which the price vector belongs to
$\Lambda_B$ ($R_{i}^{B}(\cdot)$ increases with $\lambda_i$:
\begin{equation}
\lambda_{i}^{B}=\mu_{i}^{*},\,\texttt{or}\,\lambda_{i}^{B}=\log{\frac{\alpha_i}{1-\sum_{j\neq
i}\alpha_je^{-\lambda_j}}}-\epsilon
\end{equation}

Finally, for the special case that $\mathbf{\lambda}\in\Lambda_C$,
the price of each operator $i$ is directly determined by the prices
that the other operators have selected. Namely:
\begin{equation}
\lambda_{i}^{C}=\log{\frac{\alpha_i}{1-\sum_{j\neq
i}\alpha_je^{-\lambda_j}}} \label{eq:priceareaC}
\end{equation}
Whether each operator $i$ will agree and adopt this price or not,
depends on the respective accrued revenue,
$R_{i}^{C}(\lambda_{i}^{C},\lambda_{-i})$.

In the sequel, we examine and analyze jointly the solutions of the
above optimization problems and derive the exact best response of
the $i^{th}$ operator for each vector $\mathbf{\lambda}_{-i}$ of the $I-1$ prices.

\begin{lem}
For each operator $i\in\mathcal{I}$, if $(1,\lambda_{-i})\notin \Lambda_A$, then
there is no best response price $\lambda_i^{*}$, such that
$(\lambda_i^{*},\lambda_{-i})\in\Lambda_A$. That is, operator $i$ will not select $\Lambda_A$.
\label{lem:lemma1}
\end{lem}
\begin{proof}
Given that the price vector $\mathbf{\lambda\in\Lambda_A}$, best
response price is:
\begin{equation}
\lambda_{i}^{A}=
\begin{cases}
1 & \text{if $(1,\lambda_{-i})\in \Lambda_A$},\\
l_0 + \epsilon & \text{if $(1,\lambda_{-i})\notin \Lambda_A$}.
\end{cases}
\end{equation}
where $l_0=\lambda_{i}^{C}$ is the price operator $i$ selects when $\mathbf{\lambda\in\Lambda_C}$.

If $(1,\lambda_{-i})\notin \Lambda_A$, then $l_0 +
\epsilon > 1$. Otherwise price vector $(l_0 + \epsilon, \lambda_{-i})$ will not belong to $\Lambda_A$. Therefore, $R_i^{A}(\cdot)$ is a
decreasing function at the point $\lambda_i=\lambda_0 + \epsilon$ due to
quasi-concavity property. Therefore, if $(1,\lambda_{-i})\notin \Lambda_A$, then $R_i^{C}(l_0)=R_i^{A}(l_0)>R_i^{A}(l_0+\epsilon)$ which means
that $\lambda_i^{C}$ always gives better response than
$\lambda_i^A$.
\end{proof}

\begin{lem}
Let us denote with $\mu_{i}^{*}$ the optimal solution of the
unconstraint problem $P_i^{B}$. For each operator $i\in\mathcal{I}$, if
$(\mu_{i}^{*},\lambda_{-i})\notin \Lambda_B$, then there is no best
response price $\lambda_i^{*}$, such that
$(\lambda_i^{*},\lambda_{-i})\in\Lambda_B$. \label{lem:lemma2}
\end{lem}
\begin{proof}
Given that the price vector $\mathbf{\lambda\in\Lambda_B}$, best
response price is:
\begin{equation}
\lambda_{i}^{B}=
\begin{cases}
\mu_{i}^{*} & \text{if $(\mu_{i}^{*},\lambda_{-i})\in \Lambda_B$},\\
l_0 - \epsilon & \text{if $(\mu_{i}^{*},\lambda_{-i})\notin
\Lambda_B$}.
\end{cases}
\end{equation}
and recall that $\lambda_{i}^{C}=l_0$.
If $(\mu_{i}^{*},\lambda_{-i})\notin \Lambda_B$, then $l_0 - \epsilon < \mu_{i}^{*}$. Otherwise, the price vector
$(l_0 - \epsilon, \lambda_{-i})$ cannot be in $\Lambda_B$. Therefore, $R_i^{B}(\cdot)$ is an increasing function at the point
$\lambda_i=\lambda_0 - \epsilon$ due to quasi-concavity property. Therefore, if $(\mu_{i}^{*},\lambda_{-i})\notin \Lambda_B$, then
$R_i^{C}(l_0)=R_i^{B}(l_0)>R_i^{B}(l_0-\epsilon)$ which means that $\lambda_i^{C}$ always gives better response than $\lambda_i^B$.
\end{proof}
In other words, the previous two Lemmas state that the only eligible best response for each operator $i\in\mathcal{I}$ in the price sets $\Lambda_{A}$ and $\Lambda_{B}$ are prices $\lambda_{i}^{*}=1$ and and $\lambda_{i}^{*}=\mu_{i}^{*}$ respectively.

\begin{thm}
The best response price of an operator $i$ is:
\begin{equation}
\lambda_{i}^{*}=
\begin{cases}
1, & \text{if $(1,\lambda_{-i})\in \Lambda_A$},\\
\mu_{i}^{*}, & \text{if $(\mu_{i}^{*},\lambda_{-i})\in \Lambda_B$},\\
\lambda_{i}^{C}=l_0, & \text{otherwise}.
\end{cases}
\end{equation}
\label{thm:theorem1}
\end{thm}
\begin{proof}
First we prove that $(1,\lambda_{-i})\in \Lambda_A$ and
$(\mu_{i}^{*},\lambda_{-i})\in \Lambda_B$ cannot be true at the same
time. Since $\mu_{i}^{*}$ is the optimal solution of unconstraint
$R_{i}^{B}$:
\begin{equation}
\frac{d R_i^{B}(\lambda_i)}{ d\lambda_i}=0 \Rightarrow
e^{\mu_{i}^{*}}(\mu_{i}^{*}-1)=\frac{\alpha_i}{\sum_{j \neq
i}\alpha_je^{-\lambda_j}} \label{eq:optimal-lambdaB}
\end{equation}
It is obvious that equation (\ref{eq:optimal-lambdaB}) can only hold
when $\mu_{i}^{*}>1$. Note that if $(\mu_{i}^{*},\lambda_{-i})\in
\Lambda_B$, the vector $\mathbf{\lambda}=(l,\lambda_{-i})\in \Lambda_B$ for any
price $l<\mu_{i}^{*}$. Hence, it should also hold that $\mathbf{\lambda}=(1,\lambda_{-i})\in \Lambda_B$. With a similar reasoning, when $(1,\lambda_{-i})\in \Lambda_A$, $(l,\lambda_{-i})\in \Lambda_A$ holds for any price
$l>1$ and therefore $(\mu_{i}^{*},\lambda_{-i})\in \Lambda_A$.
Also, if $(1,\lambda_{-i})\in \Lambda_A$, $\lambda_i^{C}$ cannot be a best response, because
$R_i^{A}(1)>R_i^{A}(\lambda_i^{C})=R_i^{C}(\lambda_i^{C})$.
Similarly, if $(\mu_{i}^{*},\lambda_{-i})\in \Lambda_B$,
$\lambda_i^{C}$ is not a best response.

Finally, from Lemma \ref{lem:lemma1} and Lemma \ref{lem:lemma2}, we
can say that $\lambda_i^C$ dominates all other prices if
$(1,\lambda_{-i})\notin \Lambda_A$ and
$(\mu_{i}^{*},\lambda_{-i})\notin \Lambda_B$ which concludes the proof.
\end{proof}

\subsection{Existence and Convergence Analysis of Nash Equilibriums} \label{appendixproof:ExistAndConverg}

In the previous section, we derived the best response strategy for
each player of the game $\mathcal{G_P}$. The next important steps
are \textbf{(i)} to explore the existence of Nash Equilibriums (NE)
for $\mathcal{G_P}$, and \textbf{(ii)} to study if the convergence
to them is guaranteed. In \cite{shapley}, it is proven that if the
game can be modeled as a potential game, not only the existence of
pure NEs are ensured, but also convergence to them is guaranteed
under any finite improvement path. In other words, a potential game
always converges to pure NE when the players adjust their strategies
based on accumulated observations as game unfolds. In this section,
we provide the necessary definitions for ordinal potential games, and we prove that game $\mathcal{G_P}$ belongs in this class
of games.
\begin{defn}
A game $(\mathcal{I},\lambda,\{R_i\})$ is an \textbf{ordinal
potential game}, if there is a potential function $\mathcal{P} :
[0,\lambda_{max}] \rightarrow \mathbb{R}$ such that the following
condition holds:
\begin{equation}
\mathrm{sgn}(\mathcal{P}(\lambda_i,\lambda_{-i})-\mathcal{P}(\lambda'_i,\lambda_{-i}))= \nonumber 
\end{equation}
\begin{equation}
\mathrm{sgn}(R_i(\lambda_i,\lambda_{-i})-R_i(\lambda'_i,\lambda_{-i}))
\forall i \in \mathcal{I},\,\lambda_i, \lambda'_i \in [0,\lambda_{max}] \label{eq:potential_condition}
\end{equation}
where \textnormal{sgn}($\cdot{}$) is the sign function.
\end{defn}
\begin{lem}
The game $\mathcal{G_P}$ is an ordinal potential game.
\label{lem:lemma01}
\end{lem}
\begin{proof}
We define the potential function as:
\begin{equation}
\mathcal{P}(\mathbf{\lambda})=
\begin{cases}
\sum_{j=1}^{I}{(\log{\lambda_j}-\lambda_j)},\, \text{if} \sum_{j=1}^{I}\alpha_je^{-\lambda_j}\leq 1,\\
\sum_{j=1}^{I}{(\log{\lambda_j}-\lambda_j)}-\log{(\sum_{j=1}^{I}{\alpha_je^{-\lambda_j}})}
,\,\text{else}.
\end{cases}
\end{equation}
Therefore,
\begin{eqnarray}
&\mathcal{P}&(\lambda_i,\lambda_{-i})-\mathcal{P}(\lambda_{i}^{'},\lambda_{-i})=\nonumber \\ 
&=& \begin{cases}
\log\frac{\lambda_ie^{\lambda_{i}^{'}}}{\lambda_{i}^{'}e^{\lambda_{i}}},\, \text{if}\, \lambda_i, \lambda_{i}^{'} \geq l_0 \nonumber \\

\displaystyle\log{\frac{\lambda_ie^{\lambda_{i}^{'}}(\alpha_ie^{-\lambda_{i}^{'}}+\sum_{j\neq
i}{\alpha_je^{-\lambda_j}})} {\lambda_{i}^{'}e^{\lambda_i}(\alpha_ie^{-\lambda_i}+\sum_{j\neq
i}{\alpha_je^{-\lambda_j}})}},\,\text{if}\,  \lambda_i, \lambda_{i}^{'} < l_0 \nonumber \\

 \displaystyle\log\frac{\lambda_i e^{\lambda_{i}^{'}}} {\lambda_{i}^{'}e^{\lambda_i}(\alpha_ie^{-\lambda_i}+\sum_{j\neq
i}{\alpha_je^{-\lambda_j}})} \lambda_i < l_0,\,\text{if}\, \lambda_{i}^{'} \geq l_0 \nonumber \\

\displaystyle\log\frac{\lambda_ie^{\lambda_{i}^{'}}(\alpha_ie^{-\lambda_{i}^{'}}+\sum_{j\neq
i}\alpha_je^{-\lambda_j})} {e^{\lambda_i}\lambda_{i}^{'}}
\,\text{if}\, \lambda_i \geq l_0,\,\lambda_{i}^{'} < l_0 \nonumber \\

\end{cases}
\end{eqnarray}
where $l_0=\log(\alpha_i/(1-\sum_{j\neq
i}(\alpha_je^{-\lambda_j})))$.
Moreover, using (\ref{eq:Revenue_function}),
\begin{equation}
\log{R_{i}}=
\begin{cases}
\log{\frac{\lambda_i}{e^{\lambda_i}}} + \log{\alpha_i N} & \text{if $\lambda_i \geq l_0$,}\\
\displaystyle\log{\frac{\lambda_i}{e^{\lambda_i}(\frac{\alpha_i}{e^{\lambda_i}}+\sum_{j\neq
i}{\frac{\alpha_j}{e^{\lambda_j}}})}} + \log{\alpha_i N} & \text{if
$\lambda_i < l_0$.}
\end{cases} \nonumber
\end{equation}

Now, it is straightforward to show that
$\mathcal{P}(\lambda_i,\lambda_{-i})-\mathcal{P}(\lambda'_i,\lambda_{-i})=\log{R_i(\lambda_i,\lambda_{-i})}-\log{R_i(\lambda'_i,\lambda_{-i})}$
for any operator $i \in \mathcal{I}$ and for any $\lambda_i,
\lambda'_i \in [0,\lambda_{max}]$. Since
$\log{R_i(\lambda_i,\lambda_{-i})}-\log{R_i(\lambda'_i,\lambda_{-i})}$
has always same sign as
$R_i(\lambda_i,\lambda_{-i})-R_i(\lambda'_i,\lambda_{-i})$,
condition given in (\ref{eq:potential_condition}) is satisfied, and
game $\mathcal{G_P}$ is an ordinal potential game.
\end{proof}

\subsection{Detailed Analysis of Nash Equilibriums} \label{appendixproof:DetailedAnalysisOfEquil}
In the previous section, we proved the existence of pure NE and
convergence to them. In this section, we extend our analysis further
in order to find these NEs. For the sake of simplicity, we consider
the case where all the operators have same amount of available
spectrum $W_i=W$ and hence $\alpha_i = \alpha,\,\forall i\in\mathcal{I}$.

Before starting our analysis, we rewrite constraint of set
$\Lambda_A$ given in eq. (\ref{eq:Lambda_A}) as follows:
\begin{equation}
\alpha \leq
\frac{1}{\sum_{j\in\mathcal{I}}{e^{-\lambda_j}}}=\frac{H(\{e^{\lambda_j}|j\in\mathcal{I}\})}{I}
\end{equation}
where $H(\cdot)$ is the harmonic mean function of the variables
$(e^{\lambda_1},e^{\lambda_2},\ldots,e^{\lambda_I})=(\{e^{\lambda_j}|j\in\mathcal{I}\})$:
\begin{equation}
H(\{e^{\lambda_j}|j\in\mathcal{I}\})=\frac{I}{e^{-\lambda_1}+e^{-\lambda_2}+\ldots+e^{-\lambda_I}}
\end{equation}
Therefore, if $\lambda \in\ \Lambda_A$, it is:
\begin{equation}
H(\{e^{\lambda_j}|j\in\mathcal{I}\}) \geq \alpha I
\end{equation}
Similarly, according to (\ref{eq:Lambda_B}), if $\lambda \in\
\Lambda_B$ then:
\begin{equation}
H(\{e^{\lambda_j}|j\in\mathcal{I}\}) \leq \alpha I
\end{equation}
and finally, if $\lambda \in\ \Lambda_C$:
\begin{equation}
H(\{e^{\lambda_j}|j\in\mathcal{I}\})=\alpha I
\end{equation}
Next, we define a new variable, $h$ as the natural logarithm of the
harmonic mean:
\begin{equation}
h=\log{H(\{e^{\lambda_j}|j\in\mathcal{I}\})}
\end{equation}
Note that, since $e^h$ is the harmonic mean of
$\{e^{\lambda_j}|j\in\mathcal{I}\}$, we can say that one of the
following should hold:
\begin{enumerate}
\item Every operator $i\in\mathcal{I}$ adopts the same price $\lambda_i=h$.
\item If one operator $j\in\mathcal{I}$ selects a price $\lambda_j<h$, then there must be at least
one other operator $k\in\mathcal{I}$ who will adopt a price
$\lambda_k>h$.
\end{enumerate}
Additionally, we define the variable $h_{-i}$ which is similar to
$h$ except that price of the $i^{th}$ operator is excluded. That is:
\begin{equation}
h_{-i}=\log(H(\{e^{\lambda_j}|j\in\mathcal{I}\setminus{i}\}))
\label{eq:h_i}
\end{equation}
It is obvious that if $\lambda_i > h$, then $h_{-i} < h$, if
$\lambda_i < h$, then $h_{-i} > h$, and if $\lambda_i = h$, then
$h_{-i} = h$.

\begin{lem}
If $\alpha\in A_1=(0,e/I)$, there is a unique NE
$\mathbf{\lambda}^{*}\in\Lambda_A$, with
$\mathbf{\lambda}^{*}=(\lambda_{i}^{*}=1:\,i\in \mathcal{I})$
\label{lem:lemma4}
\end{lem}
\begin{proof}
First, we prove that the NE cannot be in $\Lambda_B$ or $\Lambda_C$
($\mathbf{\lambda}^{*}\notin \Lambda_B \cup \Lambda_C$) if $\alpha
\in A_1=(0,e/I)$. Notice that, when the price vector is not in
$\Lambda_A$, $h \leq \log(\alpha I) < 1$ for given $\alpha$ values.
Therefore there exists at least one operator with price less than
one. Since $R_{i}^{B}$ is an increasing function between
$\lambda_i\in(0,1)$, operators with $\lambda_i<1$ would gain more
revenue by unilaterally increasing their prices. Therefore
$\mathbf{\lambda}^{*}$ can only be in $\Lambda_A$. According to
Theorem \ref{thm:theorem1}, given that the price vector is in
$\Lambda_A$, optimal price for any operator $i$ can only be
$\lambda_{i}^{A}=1$ if $(1,\lambda_{-i})\in \Lambda_A$. Since
$\mathbf{\lambda}^{*}=(\lambda_{i}^{*}=1:\,i\in \mathcal{I}) \in
\Lambda_A$ when $\alpha \in A_1$, it is a feasible and unique
solution.
\end{proof}

\begin{lem}
$\mu_{i}^{*}$ is always between $\frac{I}{I-1}$ and $h_{-i}$
\label{lem:lemma5}
\end{lem}
\begin{proof}
We can rewrite equation (\ref{eq:optimal-lambdaB}) as follows:
\begin{equation}
e^{\mu_{i}^{*}}(\mu_{i}^{*}-1)=\frac{e^{h_{-i}}}{I-1}
\label{eq:optimal-lambdaB-2}
\end{equation}
where $h_{-i}$ is defined in equation (\ref{eq:h_i}). Now, if
$h_{-i} < \frac{I}{I-1}$, or equivalently if $h_{-i}-1 <
\frac{1}{I-1}$, then $\lambda_{i}^{*}$ should be greater than
$h_{-i}$ in order to satisfy (\ref{eq:optimal-lambdaB-2}). Moreover,
if $\lambda_{i}^{*}>h_{-i}$, then $\lambda_{i}^{*}-1$ should be less
than $\frac{1}{I-1}$ in order to satisfy
(\ref{eq:optimal-lambdaB-2}). Therefore,
$h_{-i}<\lambda_{i}^{*}<\frac{I}{I-1}$. Similarly, if $h_{-i} \geq
\frac{I}{I-1}$, then $\frac{I}{I-1} \leq \lambda_{i}^{*} \leq
h_{-i}$, which proves the lemma.
\end{proof}

\begin{lem}
If $\alpha\in A_3=(e^{I/(I-1)},\infty)$, there is a unique NE
$\mathbf{\lambda}^{*}\in\Lambda_B$, with
$\mathbf{\lambda}^{*}=(\lambda_{i}^{*}=I/(I-1):\,i\in \mathcal{I})$
\label{lem:lemma6}
\end{lem}
\begin{proof}
First we prove that there is no NE in $\Lambda_A$ if $\alpha \geq
e/I$ (i.e. if $\alpha \in A_2 \cup A_3$). According to Theorem
\ref{thm:theorem1}, optimal price for any operator $i$ can only be
$\lambda_{i}^{A}=1$ if $(1,\lambda_{-i})\in \Lambda_A$. Otherwise
$\lambda_{i}^{C}$ dominates $\lambda_{i}^{A}$. Since
$\mathbf{\lambda}^{*}=(\lambda_{i}^{*}=1:\,i\in \mathcal{I}) \notin
\Lambda_A$ when $\alpha \in A_2 \cup A_3$, there is no NE in
$\Lambda_A$.

Secondly, we prove that there is no NE in $\Lambda_C$ if $\alpha\in
A_3=(e^{I/(I-1)},\infty)$. Recall that, when the price vector is in
$\Lambda_C$, $h = \log(\alpha I) > I/(I-1)$, which means that there
exists at least one operator with price $\lambda_i^C>I/(I-1)$ and
$\lambda_i^C \geq h$. Remember that if $\lambda_i \geq h$, then
$h_{-i} \leq h$, so $\lambda_i \geq h_{-i}$. Therefore, for an
operator $i$, $\lambda_i^C$ is greater than both $h_{-i}$ and
$I/(I-1)$. According to Theorem \ref{thm:theorem1} and Lemma
\ref{lem:lemma5}, when $(\mu_{i}^{*},\lambda_{-i})\in \Lambda_B$,
best response price of operator $i$ is $\mu_{i}^{*}$ which is
between $h_{-i}$ and $I/(I-1)$. This means that for at least one
operator, $\lambda_i^C$ is greater than $\mu_{i}^{*}$, which implies
that $(\mu_{i}^{*},\lambda_{-i})\in\Lambda_B$. This operator can
increase his revenue by reducing his price to $\mu_{i}^{*}$.
Therefore, there is no NE in $\lambda_C$ for the given $\alpha$
values, and we proved that the NE can only be in $\Lambda_B$.

Finally, we prove that the only NE is
$\mathbf{\lambda}^{*}=(\lambda_{i}^{*}=I/(I-1):\,i\in \mathcal{I})$,
if $\alpha \in A_3$. According to Lemma \ref{lem:lemma5},
$\mu_{i}^{*}$ is between $h_{-i}$ and $I/(I-1)$ for all operators.
$h$ can be greater than or less than $I/(I-1)$. If $h\geq I/(I-1)$,
unless all of the operators set their prices to $I/(I-1)$, there
exists at least one operator $i$ with price $\lambda_i$ greater than
both $h_{-i}$ and $I/(I-1)$. Hence, $\lambda_i$ is also greater than
$\mu_{i}^{*}$ and this operator can increase his revenue by reducing
his price to $\mu_{i}^{*}$. Similarly, if $h <I/(I-1)$, there exists
at least one operator with price $\lambda_i$ less than both $h_{-i}$
and $I/(I-1)$. This operator can increase his revenue by increasing
his price. If all the operators set their prices to
$\lambda_i=I/(I-1)$, the price vector is in $\Lambda_B$ and none of
the operators can increase his revenue by unilaterally changing his
price. Therefore, the only NE is
$\mathbf{\lambda}^{*}=(\lambda_{i}^{*}=I/(I-1):\,i\in \mathcal{I})$.

Hence the lemma is proved.
\end{proof}

\begin{lem}
If $\alpha \in A_2=[e/I,e^{I/(I-1)}]$, a NE can only be in
$\Lambda_C$. \label{lem:lemma7}
\end{lem}
\begin{proof}
In the proof of the Lemma \ref{lem:lemma6}, we showed that there is
no NE in $\Lambda_A$, if $\alpha\in\ A_2$. We can also prove that
there is no NE in $\Lambda_B$ for the given range of $\alpha$
values. If the price vector is in $\Lambda_B$ and
$\alpha<\frac{e^{I/(I-1)}}{I}$, then $h<I/(I-1)$. Therefore, there
is at least one operator with price $\lambda_i \leq h_{-i}$ and
$\lambda_i<I/(I-1)$, who can increase his revenue by increasing his
price. So we conclude that, if $\alpha \in A_2=[e/I,e^{I/(I-1)}]$,
there is no NE in $\Lambda_A$ or $\Lambda_B$.
\end{proof}

\begin{lem}
If $\alpha \in A_2=[e/I,e^{I/(I-1)}]$,
$\mathbf{\lambda}^{*}\in\Lambda_C$ with
$\mathbf{\lambda}^{*}=(\lambda_{i}^{*}=\log(I\alpha):\,i\in
\mathcal{I})$ is a NE. \label{lem:lemma8}
\end{lem}
\begin{proof}
When all the operators set the same price
$\lambda_{i}^{*}=\log(I\alpha)$ and $\alpha \in
A_2=[e/I,e^{I/(I-1)}]$, $\log(I\alpha)$ is between 1 and
$\mu_{i}^{*}$ for all operators (this can be verified through
equation (\ref{eq:optimal-lambdaB})). Therefore, for any operator i,
$(1,\lambda_{-i})\notin\Lambda_A$ and
$(\mu_{i}^{*},\lambda_{-i})\notin\Lambda_B$. Then, according to
Theorem \ref{thm:theorem1}, best response price is $\lambda_i^C$
which is equal to $\log(I\alpha)$. Hence, no operator can gain more
revenue by unilaterally changing his price, and
$\mathbf{\lambda}^{*}=(\lambda_{i}^{*}=\log(I\alpha):\,i\in
\mathcal{I})$ is a NE.
\end{proof}
\ \par Finally we analyze the NE for boundary values of $A_2$, i.e.
for $\alpha=e/I$ and $\alpha=e^{I/(I-1)}/I$. In the previous lemma,
it is proven that
$\mathbf{\lambda}^{*}=(\lambda_{i}^{*}=\log(I\alpha):\,i\in
\mathcal{I})$ is a NE if $\alpha\in A_2$. We can also prove that it
is the only NE for these boundary values. In Lemma \ref{lem:lemma7},
it is proven that any NE is in $\Lambda_C$ for these $\alpha$
values. So, when $\alpha=e/I$, $h=\log(I\alpha)=1$, which means that
unless all of the users set their prices to one, there exists some
operators with $\lambda_i<1$. These operators would gain more
revenue by setting their prices to one. Hence, the only NE is
$\lambda_i^{*}=\log(I\alpha)=1$. Similarly, when
$\alpha=e^{I/(I-1)}/I$, $h=\log(I\alpha)=I/(I-1)$. This means that
unless  all of the users set their prices to $I/(I-1)$, there exists
some operators with $\lambda_i$ greater than both $h_{-i}$ and
$M/(M-1)$. These operators would gain more revenue by reducing their
prices. Hence, the only NE is $\lambda_i^{*}=\log(I\alpha)=I/(I-1)$.

We also show that when $\alpha \in (e/I,e^{I/(I-1)})$, there can be
infinitely many NEs, all in $\Lambda_C$, via numerical simulations.
For different initial price settings, the game converges to
different NE.

\begin{thm}
The game $\mathcal{G_P}$ attains a pure NE which depends on the
value of parameter $\alpha$ as follows:
\begin{itemize}
\item If $\alpha\in A_1=(0,e/I)$, there is a unique NE
$\mathbf{\lambda}^{*}\in\Lambda_A$, with
$\mathbf{\lambda}^{*}=(\lambda_{i}^{*}=1:\,i\in \mathcal{I})$ and
respective market equilibrium $\mathbf{x}^{*}\in X_A$.
\item If $\alpha\in A_3=(e^{\frac{I}{I-1}}/I,\infty)$, there is a unique NE $\mathbf{\lambda}^{*}\in\Lambda_B$, with
$\mathbf{\lambda}^{*}=(\lambda_{i}^{*}=\frac{I}{I-1}:\,i\in
\mathcal{I})$, which induces a respective market equilibrium
$\mathbf{x}^{*}\in X_B$.
\item If $\alpha\in A_2=[e/I, e^{\frac{I}{I-1}}/I]$, there exist infinitely many
NEs, $\mathbf{\lambda}^{*}\in\Lambda_C$, and each one of them yields
a respective market stationary point $\mathbf{x}^{*}\in X_C$.
\end{itemize}
\end{thm}
\begin{proof}
Lemma \ref{lem:lemma01} proves that the game $\mathcal{G_P}$ is a
finite ordinal potential game. Therefore, it always attains a pure
NE. Lemma \ref{lem:lemma4} proves the case for $\alpha\in
A_1=(0,e/I)$. Lemma \ref{lem:lemma6} proves the case for $\alpha\in
A_3=(e^{\frac{I}{I-1}}/I,\infty)$. The case for $\alpha\in A_2=[e/I,
e^{\frac{I}{I-1}}/I]$ is proven in Lemma \ref{lem:lemma7} and Lemma
\ref{lem:lemma8}.
\end{proof}

%%%--------\end{subappendices}

% that's all folks
\end{document}